\documentclass[twoside]{article}
\usepackage[accepted]{aistats2018}
\usepackage{hyperref}       
\usepackage{url}            
\usepackage{booktabs}       
\usepackage{amsfonts}       
\usepackage{nicefrac}       
\usepackage{microtype}      
\hypersetup{colorlinks=true,citecolor=blue,linkcolor=blue}

\usepackage{longtable}
\usepackage{multirow}
\usepackage{amsmath}
\usepackage{amsthm}
\usepackage{amssymb}
\usepackage{booktabs}
\usepackage{algorithm}
\usepackage{color}
\usepackage{subfig} 
\usepackage{graphicx}
\usepackage{algpseudocode}
\usepackage{lscape}

\usepackage{natbib}
\setcitestyle{authoryear,open={(},close={)}}
\usepackage{bbm}

\newtheorem{theorem}{Theorem}[section]
\newtheorem{corollary}[theorem]{Corollary}
\newtheorem{lemma}[theorem]{Lemma}
\newtheorem{definition}[theorem]{Definition}

\newtheorem{example}[theorem]{Example}

\makeatletter
\newcommand*{\RN}[1]{\expandafter\@slowromancap\romannumeral #1@}
\makeatother

\newcommand {\eq} [1] {\begin{equation}\label{#1}}
\newcommand {\en} {\end{equation}}
\newcommand {\cA}       {{\cal A}}

\newcommand {\R}        {{\mathbb R}}

\newcommand {\mat}      [1] {\left[\begin{array}{#1}}
\newcommand {\rix}          {\end{array}\right]}

\newcommand {\diag}     {\mathop{\rm diag}\nolimits}

\newcommand {\rank}     {\mathop{\rm rank}\nolimits}

\newcommand {\range}    {\mathop{\rm range}\nolimits}

\DeclareMathOperator*{\E}{\mathbb{E}}

\newcommand{\eps}{\epsilon}
\renewcommand{\E}{{ \bf{E}}}
\DeclareMathOperator{\OPT}{OPT}

\DeclareMathOperator{\poly}{poly}

\DeclareMathOperator{\nnz}{\mathtt{nnz}}

\newcommand{\normF}[1]{{\| #1 \|}_F}
\newcommand{\norm}[1]{\lVert #1 \rVert_2}

\newcommand\ty{\tilde{y}}
\newcommand\hA{\hat{A}}

\newcommand\hb{\hat{b}}

\newcommand\hS{\hat{S}}

\newcommand\Iden{{I}}

\DeclareMathOperator*{\argmin}{argmin}
\newcommand\twomat[2]{\left[\begin{smallmatrix} #1 \\ #2 \end{smallmatrix} \right] }
\DeclareMathOperator{\sd}{\mathtt{sd}}

\usepackage{etoolbox}

\makeatletter

\newcommand{\define}[4][ignore]{%
  \ifstrequal{#1}{ignore}{}{
  \@namedef{thmtitle@#2}{#1}}%
  \@namedef{thm@#2}{#4}%
  \@namedef{thmtypen@#2}{lemma}%
  \newtheorem{thmtype@#2}[theorem]{#3}%
  \newtheorem*{thmtypealt@#2}{#3~\ref{#2}}%
}

\newcommand{\state}[1]{%
  \@namedef{curthm}{#1}
  \@ifundefined{thmtitle@#1}{
  \begin{thmtype@#1}
    }{
  \begin{thmtype@#1}[\@nameuse{thmtitle@#1}]
  }
    \label{#1}
    \@nameuse{thm@#1}
  \end{thmtype@#1}
  \@ifundefined{thmdone@#1}{
  \@namedef{thmdone@#1}{stated}%
  }{}
}

\newcommand{\restate}[1]{%
  \@namedef{curthm}{#1}
  \@ifundefined{thmtitle@#1}{
    \begin{thmtypealt@#1}
    }{
  \begin{thmtypealt@#1}[\@nameuse{thmtitle@#1}]
  }
    \@nameuse{thm@#1}
  \end{thmtypealt@#1}
  \@ifundefined{thmdone@#1}{
  \@namedef{thmdone@#1}{stated}%
  }{}
}

\newcommand{\thmlabel}[1]{
  \@ifundefined{thmdone@\@nameuse{curthm}}{\label{#1}
    }{\tag*{\eqref{#1}}}
}
\makeatother

\begin{document}

\runningauthor{Diao, Song, Sun, Woodruff}
%

%

\twocolumn[

\aistatstitle{Sketching for Kronecker Product  Regression and P-splines}

\aistatsauthor{
  Huaian Diao \\
  \texttt{hadiao@nenu.edu.cn}\\
  \And
  Zhao Song \\
  \texttt{zhaos@utexas.edu}\\
  \And 
  Wen Sun \\
  \texttt{wensun@cs.cmu.edu} \\
  \And
  David P. Woodruff \\
  \texttt{dwoodruf@cs.cmu.edu}
}

\aistatsaddress{Northeast Normal University \And Harvard U \& UT-Austin \And Carnegie Mellon University \And Carnegie Mellon University } 
]

\begin{abstract}
   {\sc TensorSketch} is an oblivious linear sketch introduced in \citep{p13} and later used in \citep{pp13} in the context of SVMs for polynomial kernels. It was shown in \citep{anw14} that {\sc TensorSketch} provides a subspace embedding, and therefore can be used for canonical correlation analysis, low rank approximation, and principal component regression for the polynomial kernel. We take {\sc TensorSketch} outside of the context of polynomials kernels, and show its utility in applications in which the underlying design matrix is a Kronecker product of smaller matrices. This allows us to solve Kronecker product regression and non-negative Kronecker product regression, as well as regularized spline regression. Our main technical result is then in extending {\sc TensorSketch} to other norms. That is, {\sc TensorSketch} only provides input sparsity time for Kronecker product regression with respect to the $2$-norm. We show how to solve Kronecker product regression with respect to the $1$-norm in time sublinear in the time required for computing the Kronecker product, as well as for more general $p$-norms.

\end{abstract}

\section{INTRODUCTION}
\label{Sec:intro}
In the overconstrained least squares regression problem, we are given an $n \times d$ matrix
$A$ called the ``design matrix'', $n \gg d$, and
an $n \times 1$ vector $b$, and the goal is to find an $x$ which minimizes $\|Ax-b\|_2$.
There are many variants to this problem, such as regularized versions of the problem in which one
seeks an $x$ so as to minimize $\|Ax-b\|_2^2 + \|Lx\|_2^2$ for a matrix $L$, or regression problems
which seek to minimize more robust loss functions, such as $\ell_1$-regression $\|Ax-b\|_1$. 

In the era of big data, large scale matrix computations have attracted considerable interest. To obtain
reasonable computational and time complexities for large scale matrix computations, a number of randomized
approximation algorithms have been developed.
For example, in  \citep{w14}, it was shown how to output a vector
$x' \in \mathbb{R}^d$ for which $\|Ax'-b\|_2 \leq (1+\epsilon)\min_{x \in \mathbb{R}^d}\|Ax-b\|_2$ in
$\nnz(A) + \poly(d/\epsilon)$ time, where $\nnz(A)$ denotes the number of non-zero entries of matrix $A$.
We refer the reader to the recent surveys \citep{kv09,m11,w14} for a detailed
treatment of this topic. 

In this work we focus on regression problems for which the design matrix is a 
{\it Kronecker product matrix}, that is, it has the form $A \otimes B$ for $A \in \mathbb{R}^{n_1 \times d_1}$
and $B \in \mathbb{R}^{n_2 \times d_2}$. Also, $b \in \mathbb{R}^{n_1 n_2}$, and one seeks to solve the problem
$\min_{x \in \mathbb{R}^{d_1 d_2}} \|(A \otimes B)x-b\|_2$ in the standard setting, which can also be generalized
to regularized and robust variants. One can also ask the question when the design matrix is a Kronecker
product of more than two matrices.

Kronecker product matrices have many applications in statistics, linear system theory, signal processing,
photogrammetry, multivariate data fitting, 
etc.; see \citep{golubVanLoan2013Book,vanLoanKron,vanbook}. The linear least squares problem involving the Kronecker product arises in many applications,
such as structured linear total least norm on blind deconvolution problems \citep{oh2005}, constrained linear least squares
problems with a Kronecker product structure, the bivariate problem of surface fitting and multidimensional density smoothing \citep{eilers2006multidimensional}.

One way to solve Kronecker product regression is to form the matrix $C = A \otimes B$ explicitly, where $A\in\mathbb{R}^{n_1\times d_1}, B\in\mathbb{R}^{n_2\times d_2}$, and then
apply a randomized algorithm to $C$. However, this takes at least $\nnz(A) \cdot \nnz(B)$
time and space. It is natural to ask if it is possible to 
solve the Kronecker product regression problem in time faster than computing $A \otimes B$. This is in fact
the case, as Fausett and Fulton \citep{Fausett1994} show that one can solve Kronecker product regression in
$O(n_1 d_1^2 + n_2 d_2^2)$ time; indeed, the solution vector $x = \textrm{vec}((B^{\perp})^\top D^{-1}A^{\perp})$, where
$D = \textrm{vec}(b)$ and vec$(E)$ for a matrix $E$ denotes the operation of stacking the columns of $E$ into
a single long vector. While such a computation does not involve computing $A \otimes B$, it is more
expensive than what one would like.

A natural question is if one can approximately solve Kronecker product
regression in $\nnz(A) + \nnz(B) + \poly(d/\epsilon)$ time, which, up to the $\poly(d/\epsilon)$ term, would
match the description size of the input. Another natural question is
Kronecker product regression with regularization.
Such regression problems arise frequently in the context of splines \citep{eilers2006multidimensional}. Finally, what about Kronecker product
regression with other, more robust, norms such as the $\ell_1$-norm?
\vspace{-2mm}
\subsection{Our Contributions}
\vspace{-2mm}
We first observe that the random linear map {\sc TensorSketch}, introduced in the context of problems for
the polynomial kernel by \cite{p13}, \cite{pp13}, is exactly suited for this task. Namely,
in \citep{anw14} it was shown that for a $d$-dimensional subspace $C$ of $\mathbb{R}^n$, represented as
an $n \times d$ matrix, there is a distribution on linear maps $S$ with $O(d^2/\epsilon^2)$ rows
such that with constant probability,
simultaneously for all $x \in \mathbb{R}^d$, $\|SCx\|_2 = (1 \pm \epsilon)\|Cx\|_2$. That is, $S$
provides an {\it Oblivious Subspace Embedding (OSE)} for the column span of $C$. Further, it is known that if $b$
is an $n$-dimensional vector, then one has that
$\|S[C,b]x\|_2 = (1 \pm \epsilon)\|[C,b]x\|_2$ for all $x \in \mathbb{R}^{d+1}$. 
Consequently,
to solve the regression problem $\min_{x \in \mathbb{R}^d}\|Cx-b\|_2$, one can instead solve the much
smaller problem $\min_{x \in \mathbb{R}^d}\|SCx-Sb\|_2$, and the solution $x'$ to the latter problem
will satisfy $\|Cx'-b\|_2 \leq (1+\epsilon)\min_{x \in \mathbb{R}^d}\|Cx-b\|_2$. 

Importantly, if $n = n_1 \cdot n_1$
and there is a basis for the column span of
$C$ of the form $A_1 \otimes A_1, A_2 \otimes A_2, \ldots, A_d \otimes A_d$,
then given $A_1, \ldots, A_d$, it holds that $SC$ can be computed in $\nnz(A)$ time, where $A$ is
the $n_1 \times d$ matrix whose $i$-th column is $A_i$. Further, given a vector $b$ with $\nnz(b)$
non-zero entries, one can compute $Sb$ in $\nnz(b)$ time. Thus, one obtains a $(1+\epsilon)$-approximate
solution to the regression $\min_{x \in \mathbb{R}^d}\|Cx-b\|_2$ in
$\nnz(A) + \nnz(b) + \poly(d/\epsilon)$ time. 

While not immediately useful for our problem, we show that via simple modifications, the claim
about {\sc TensorSketch} 
above can be generalized to the case when there is a basis for the column span of $C$ of the
form $A_i \otimes B_j$ for arbitrary vectors $A_1, \ldots, A_{d_1} \in \mathbb{R}^{n_1}$ and
vectors $B_1, \ldots, B_{d_2} \in \mathbb{R}^{n_2}$. That is, we observe that
in this case $SC$ can be computed in $\nnz(A) + \nnz(B)$ time, where $A$ is the $n_1 \times d_1$
matrix whose $i$-th column is $A_i$, and $B$ is the $n_2 \times d_2$ matrix whose $i$-th column is $B_i$.
In this case $C = A \otimes B$, which is exactly the case of Kronecker product regression. Using the
above connection to regression, we obtain an algorithm for Kronecker product regression in
$\nnz(A) + \nnz(B) + \nnz(b) + \poly(d_1 d_2/\epsilon)$ time. Using the fact that
$\|SCx-Sb\|_2 = (1 \pm \epsilon)\|Cx-b\|_2$ for all $x$, we have in particular that this holds for
all non-negative $x$, and so also obtain the same reduction in problem size for
non-negative Kronecker product regression, which occurs often in image and text
data; see e.g., \citep{chen2012}.

The above observation allows us to extend many existing variants of least squares regression to the
case when $C$ is a Kronecker product of two matrices. For example, the results in \citep{acw16a} for the ridge
regression problem $\min_{x \in \mathbb{R}^d}\|Ax-b\|_2^2 + \lambda \|x\|_2^2$
immediately hold when $C$ is a Kronecker product matrix, since the conditions needed
for the oblivious embedding in \citep{acw16a} come down to a subspace embedding and an approximate matrix
product condition, both of which are known to hold for {\sc TensorSketch} \citep{anw14}. More interestingly
we are able to extend the results in \citep{acw16a} for Kronecker ridge regression to the case when the regularizer
is a general matrix, namely, to the problem $\min_{x \in \mathbb{R}^d} \|Ax-b\|_2^2 + \lambda \|Lx\|_2^2$, where
$L$ is an arbitrary matrix. Such problems occur in the context of spline regression \citep{eilers96,eilers2006multidimensional,eilers15}.
The number of rows of {\sc TensorSketch} depends on a generalized
notion of statistical dimension depending on the generalized singular values $\gamma_i$ of $[A; L]$, and
may be much smaller than $d_1$ or $d_2$.
In \citep{acw16a}, only results for $L$ equal to the identity were obtained.

Finally, our main technical result is to extend our results to
least absolute deviation Kronecker product regression $\min_{x \in \mathbb{R}^d}\|Cx-b\|_1$, which, since it involves
the $1$-norm, is often considered more robust than least squares regression \citep{rousseeuw2005robust} and has been widely used in applications such as computer vision \citep{zheng2012practical}. We in fact extend this to general $p$-norms
but focus the discussion on $p = 1$. We give the first algorithms
for this problem that are faster than computing $C = A \otimes B$ explicitly, which would take at least
$\nnz(A) \cdot \nnz(B) \geq n_1 n_2$ time. Namely, and for simplicitly focusing on the case when $n_1 = n_2$,
for which the goal is to do better than $n_1^2$ time, 
we show how to output an $x \in \mathbb{R}^{d_1 \times d_2}$ for which
$\|Cx'-b\|_1 \leq (1+\epsilon) \min_{x \in \mathbb{R}^{d_1 d_2}}\|Cx-b\|_1$ in time
$n^{3/2} \poly(d_1 d_2/\epsilon)$. While this is more expensive than solving Kronecker product least squares,
the $1$-norm may lead to significantly more robust solutions. From a technical perspective, {\sc TensorSketch}
when applied to a vector actually destroys the $1$-norm of a vector, preserving only its $2$-norm, so new
ideas are needed. We show how to use multiple {\sc TensorSketch} matrices to implicitly
obtain very crude estimates to the so-called $\ell_1$-leverage scores of $C$, which can be interpreted as
probabilities of sampling rows of $C$ in order to obtain an $\ell_1$-subspace embedding of the column span
of $C$ (see, e.g., \citep{w14}). However, since $C$ has $n_1 n_2$ rows, we cannot even afford to write down the $\ell_1$-leverage
scores of $C$. We show how to implicitly represent such leverage scores and to sample from them without ever
explicitly writing them down. Balancing the phases of our algorithm leads to our overall time.
\vspace{-2mm}
\subsection{Notation}
\vspace{-2mm}
We consider Kronecker Product of $q$ 2-d matrices $A_1\otimes A_2\otimes ...\otimes A_q$, where each $A_i\in\mathbb{R}^{n_i\times d_i}$. We denote $\mathcal{A} = A_1\otimes A_2\otimes ...\otimes A_q$, $n = \prod_{i=1}^q n_i$, and $d = \prod_{i=1}^q d_i$.  We denote $[n]$ as the set $\{ 1,2,3, \cdots ,n \}$. The $\ell_p$ norm for a vector $x\in\mathbb{R}^d$ is defined as $\|x\|_p = (\sum_{i=1}^d | x_i |^p)^{1/p}$, where $x_i$ stands for the $i$'th entry of the vector $x$.   For any matrix $M$, we use $M_{i,*}$ to represent the $i$'th row of $M$ and $M_{*,j}$ as the $j$'th column of $M$.


We define a \emph{Well-Conditioned Basis} and \emph{Statistical Dimension} as follow (similar definitions can be found in \cite{c05,ddhkm09,sw11,mm13,swz17,swz17b} and \cite{acw16a}):
\begin{definition} [Well-Conditioned Basis]
\label{def:wcb}
Let $A$ be an $n\times m$ matrix with rank $d$, let $p\in[1,\infty)$, and let $\|\cdot\|_q$ be the dual norm of $\|\cdot\|_p$, i.e., $1/p+1/q = 1$. Then an $n\times d$ matrix $U$ is an $(\alpha,\beta,p)$-well-conditioned basis for the column space of $A$, if the columns of $U$ span the column space of $A$ and (1) $\|U\|_p\leq \alpha$, and (2) for all $z\in\mathbb{R}^d$, $\|z\|_q\leq \beta\|Uz\|_p$.
\end{definition}

Consider the classic Ridge Regression: $\min_x\|Ax - b\|^2_2 + \lambda\|x\|_2^2$. The \emph{Statistical Dimension} is defined as:  
\begin{definition} [Statistical Dimension]
Let $A$ be an $n\times m$ matrix with rank $d$ and singular values $\sigma_i,i\in [d]$. For $\lambda \geq 0$, the statistical dimension $\sd_\lambda(A)$ is defined as the quantity $\sd_\lambda(A) = \sum_{i\in[d]}1/(1+\lambda/\sigma_i^2)$.
\end{definition}


\vspace{-2mm}
\section{BACKGROUND:{\sc TensorSketch} }\label{sec:background}
\vspace{-2mm}
We briefly introduce \textsc{TensorSketch} \citep{p13,anw14} and how to apply \textsc{TensorSketch} to the Kronecker product of multiple matrices efficiently without explicitly computing the tensor product.\footnote{We refer readers to \citep{anw14} for more details about \textsc{TensorSketch}.}

We want to find a oblivious subspace embedding $S$ such that for any $x\in\mathbb{R}^{n}$, we have $\|S\mathcal{A}x\|_2 = (1\pm \epsilon)\|\mathcal{A}x\|_2$, where the notation $a=(1\pm\epsilon b)$ stands for $(1-\epsilon)b \leq a\leq (1+\epsilon) b$, for any $a,b\in\mathbb{R}$. Consider the $(i_1,i_2,...,i_q)$'th column of $\mathcal{A}$ ($i_j\in [d_j]$): $A_{1{_{*,i_1}}}\otimes A_{2_{*,i_2}}\otimes \cdots \otimes A_{{q}_{*,i_q}}$.  Assume the sketching target dimension is $m$. {\sc TensorSketch} is defined using $q$ 3-wise indepedent hash functions $h_i: [n_i]\to [m]$, and $q$ 4-wise independent sign functions $s_i: [n_i]\to \{-1,1\}$, $\forall i\in[q]$. Define hash function $H: [n_1] \times [n_2] \times \cdots \times [n_q]\to [m]$ as $H(i_1,i_2,...,i_q) = ((\sum_{k=1}^q h_k(i_k))\mod m) $, and sign function $S: [n_1] \times [n_2] \times \cdots \times [n_q]\to \{-1,1\}$ as $S(i_1,i_2,..., i_q) =\prod_{k=1}^q s_k(i_k)$.  Applying {\sc Tensorsketch} to the Kronecker product of vectors $A_{1_{*,i_1}}, ..., A_{q_{*,i_q}}$ is equivalent to applying {\sc CountSketch} \citep{ccf04} defined with $H$ and $S$ to the vector $(A_{1_{*,i_1}}\otimes ...\otimes A_{q_{*,i_q}})$.  To apply {\sc Tensorsketch} to $\mathcal{A}$, we just need to apply {\sc CountSketch} defined with $H$ and $S$ to the columns of $\mathcal{A}$ one by one. 

Applying {\sc tensorsketch} to  the Kronecker product of $(A_{1_{*,i_1}}, ...,  A_{q_{*,i_q}})$ naively would require at least $O(n)$ time. \cite{p13} shows that one can apply {\sc tensorsketch} to the Kronecker product of $(A_{1_{*,i_1}}, ...,  A_{q_{*,i_q}})$ without explicitly computing the Kronecker product of these vectors using the Fast Fourier Transformation. Particularly, \cite{pp13} show that one only needs $O( \sum_{j=1}^q \nnz(A_{j_{*,i_j}})+ qm\log (m))$ time to compute $S(A_{1_{*,i_1}}\otimes ...\otimes A_{q_{*,i_q}})$ where $A_{i_{*,j}}$ stands for the $j$'th column of $A_i$. As $\mathcal{A}$ has $d$ columns, computing $S\mathcal{A}$ takes $O(d(\sum_{i=1}^q \nnz(A_i)) + dqm\log(m))$ time, which is much smaller than $O(\prod_{i=1}^q n_i)$, which is the least amount of time one needs for explicitly computing $\mathcal{A}$. In the rest of  the paper, we  assume that we compute $S\mathcal{A}$ using the above efficient procedure without explicitly computing $A_1\otimes \ldots \otimes A_q$.

\section{TENSOR PRODUCT LEAST SQUARES REGRESSION}
\label{Sec:TensorRegression}

Consider the tensor product least squares regression problem  $\min_{x}\|\mathcal{A}x-b\|_2$ where $\mathcal{A} \in \mathbb{R}^{n\times d}$ and $b\in \R^{n}$.  Let $S \in \R^{m\times n}$ be the matrix form of {\sc TensorSketch} of Section \ref{sec:background}. We propose a {\sc TensorSketch} -type Algorithm \ref{alg:tensorregression} for the tensor product regression problem. The following theorem shows that the solution obtained from Alg.~\ref{alg:tensorregression} is a good approximation of the optimal solution of the original tensor product regression. 

Let us define $\OPT$ to be the optimal cost of the optimization problem, e.g., $\OPT  = \min_{x} \| \cA x- b \|_2$. The following theorem shows that Alg.~\ref{alg:tensorregression} computes a good approximate solution.

\begin{algorithm}[tb]
 \caption{Tensor product regression}
  \label{alg:tensorregression}
  \begin{algorithmic}[1]
  \Procedure{TRegression}{$A,b,\epsilon,\delta$}
    \State $m \leftarrow (d_1d_2\cdots d_q +1)^2 (2+3^q)/(\eps^2\delta)$
    \State Choose $S$ to be an $m \times (n_1n_2\cdots n_q)$ {\sc TensorSketch} matrix
    \State Compute $S(A_1\otimes A_2\otimes \cdots \otimes A_q)$ and $Sb$
    \State $ \widetilde { x} \leftarrow \min_{x}\|S(A_1\otimes A_2\otimes \cdots \otimes A_q)x-Sb\|_2$
    \State \Return $\widetilde{x}$
\EndProcedure

    \end{algorithmic}
\end{algorithm}

\define{them:norm_two_reg}{Theorem}{ \rm{(Tensor regression)}
	Suppose $ \widetilde { x}$ is the output of Algorithm \ref{alg:tensorregression} with {\sc tensorsketch} $S\in \mathbb{R}^{m\times n}$, where $m= 8(d_1d_2\cdots d_q +1)^2 (2+3^q)/(\eps^2\delta)$. Then the following approximation
	$
\|(A_1\otimes A_2\otimes \cdots \otimes A_q) \widetilde { x}-b\|_2	 
\leq (1+\epsilon) \OPT,	
	$	holds with probability at least $1-\delta $.
}
\state{them:norm_two_reg}
The proof of Theorem~\ref{them:norm_two_reg} can be found in Appendix ~\ref{sec:proof_norm_two}.
Theorem~\ref{them:norm_two_reg} shows that we can achieve an $\epsilon$-close solution by solving a much smaller regression problem with a number of samples of order $O(\poly(d/\epsilon))$, which is independent of the large dimension $n$. Using the technique we introduced in Sec.~\ref{sec:background}, we can also compute $S\cA$ without explicitly computing the tensor product.

We can extend Theorem~\ref{them:norm_two_reg} to the nonnegative tensor product regression problem
$
	\min_{x \geq 0 }\|(A_1\otimes A_2\otimes \cdots \otimes A_q)x-b\|_2,
$
where $A_i\in \R^{n_i \times d_i},\, i=1,\ldots, q$ and $b\in \R^{n_1n_2\cdots n_q}$. Suppose $ x$ is the optimal solution. Similarly, let $S \in \R^{m\times (n_1 n_2\cdots n_q)}$ be the matrix form of {\sc TensorSketch} of Section \ref{sec:background}.  If $ \widetilde { x} $ is the optimal solution  to
 $
	\min_{x \geq 0}\|S(A_1\otimes A_2\otimes \cdots \otimes A_q)x-Sb\|_2.
$, we have the following:

%
%
%
%
%
%
%
%
%
\define{them:norm_two_reg_positive}{Corollary}{ \rm{(Sketch for tensor nonnegative regression)}
	Suppose $\tilde{x} = \min_{x\geq 0} \|S\mathcal{A}x-Sb\|_2 $ with {\sc tensorsketch} $S\in \mathbb{R}^{m\times n}$, where $m=8 (d_1d_2\cdots d_q +1)^2 (2+3^q)/(\eps^2\delta)$.  Then the following approximation
	$
\|(A_1\otimes A_2\otimes \cdots \otimes A_q) \widetilde { x}-b\|_2	 
\leq  (1+\epsilon)  \OPT	
	$	holds with probability at least $1-\delta $, 
	where $\OPT = \min_{x\geq 0}
\left\|(A_1\otimes A_2\otimes \cdots \otimes A_q)  { x}-b\right\|_2$.
}
\state{them:norm_two_reg_positive}
 The proof of Corollary~\ref{them:norm_two_reg_positive} can be found in Appendix~\ref{sec:proof_norm_two}.

\vspace{-2mm}
\section{ P-SPLINES }
\vspace{-2mm}
B-splines are local basis functions, consisting of low degree (e.g., quadratic, cubic) polynomial segments. The positions where the segments join are called the knots. B-splines have local support and are thus suitable for smoothing and interpolating data with complex patterns. 
Unfortunately, control over smoothness is limited: one can only change the number and positions of the knots. If there are no reasons to assume that smoothness is non-uniform, the knots will be equally spaced and the only tuning parameter is their (discrete) number. 
In contrast P-spline \citep{eilers96} equally spaces B-splines, discards the derivative completely, and controls smoothness by regularizing the sum of squares of differences of coefficients. Specifically Eilers and Marx \cite{eilers96} proposed the P-spline recipe: (1) use a (quadratic or cubic) B-spline basis with a large number of knots, say 10-50; (2) introduce a penalty on (second or third order) differences of the B-spline coefficients; (3) minimize the resulting penalized likelihood function; (4) tune smoothness with the weight of the penalty, using cross-validation or AIC to determine the optimal weight.

We give a brief overview of B-Splines and P-Splines below. Let $b$ and $u$, each vectors of length $n$, represent the observed and explanatory variables, respectively. Once a set of knots is chosen, the B-spline basis $A$ follows from $u$. If there are $d$ basis functions then $A$ is $n \times  d$. In the case of normally distributed observations the model is $b= A x + e$, with independent errors $e$. In the case of B-spline regression the sum of squares of residuals $\left\|b- A x \right\|_2$ is minimized and the normal equations $A^\top A \hat x = A^\top b$ are obtained; the explicit solution  $\hat x = (A^\top A)^{-1} A^\top b$ results. The P-spline approach minimizes the penalized least-squares function
\begin{equation}\label{eq ridge1}
	\left\|b- A x \right\|_2^2+\lambda \|Lx\|_2^2,
\end{equation}
where $L \in  \R^{p\times n}$ is a matrix that forms differences of order $\ell $, i.e., $L_\ell x=\Delta^\ell x $. Examples of this matrix, for $\ell =1$ and $\ell =2$ are :
\begin{align*}
		L_1=\begin{bmatrix}
	-1 & 1 & 0 &0 \\
	0 &-1 & 1 & 0\\
	0 & 0 &-1 &1
\end{bmatrix}, L_2=\begin{bmatrix}
	1 & -2 &1 & 0 &0 \\
	0 & 1&-2 & 1 & 0\\
	0 & 0 & 1&-2 &1
\end{bmatrix}.
\end{align*}

The parameter  $\lambda$ determines the influence of the penalty. If $\lambda$  is zero, we are back to B-spline regression; increasing $\lambda$  makes  $\hat x$, and hence  $\hat b = A \hat x$, smoother. 

Let
$x^*$ denote $\argmin_{x\in\R^d} \norm{Ax-b}^2 + \lambda\norm{L x}^2$, and $\OPT$ denote $\norm{Ax^*-b}^2 + \lambda\norm{Lx^*}^2$.
In general $x^* = (A^\top A + \lambda L^\top L)^{-1} A^\top b = A^\top (A A^\top + \lambda LL^\top )^{-1}b$,
so $x^*$ can be found in $O(\nnz(A) \min(n, d))$ time using an iterative method (e.g., LSQR). 
Our first goal in this section is to design faster algorithms that find an approximate
$\widetilde{x}$ in the following sense:
\begin{equation}
\label{eq:goal ridge regression}
\norm{A\widetilde{x}-b}^2 + \lambda\norm{L \widetilde{x}}^2\le (1+\eps)\OPT. \vspace{-4mm}
\end{equation}
\vspace{-5mm}
\subsection{Sketching for P-Spline}
\vspace{-2mm}
We first introduce a new definition of \emph{Statistical Dimension} that extends the statistical dimension defined for Ridge Regression (i.e., $L$ is an identity matrix in Eq.~\ref{eq ridge1}) \citep{acw16a} to P-Spline regression. 

The problem (\ref{eq ridge1}) can also be analyzed  by 
generalized singular value decomposition
(GSVD) \cite{golubVanLoan2013Book}. For  matrices
$A \in \R^{n \times d}$ and $L \in \R^{p \times d}$ with
$\rank(L)=p$ and $\rank\left(\mat{c} A \\
L \rix\right) =d$, the GSVD of $(A,L)$ is given by the pair of factorizations
$
A=U\begin{bmatrix}\Sigma& { 0}_{p\times (n-p)}\cr { 0}_{(n-p)\times p}& I_{d-p}
\end{bmatrix}RQ^\top  \quad \hbox{and} \quad
L=V\begin{bmatrix}\Omega  &{
0}_{p \times (n-p)}\end{bmatrix}RQ^\top, 
$
where $U \in \R^{m\times n}$ has orthonormal columns,
$V \in \R^{p \times p}$, $Q\in \R^{d \times d} $ are orthogonal, $R\in\mathbb{R}^{d\times d}$ is
upper triangular and nonsingular, and $\Sigma$ and $\Omega$ are
$p\times p$ diagonal matrices:
$\Sigma=\diag(\sigma_1,\sigma_2,\ldots,\sigma_{p})$ and
$\Omega=\diag(\mu_1,\mu_2,\ldots,\mu_p)$ with
$
0 \leq \sigma_1 \leq \sigma_2 \leq \ldots \leq \sigma_p < 1
\quad \hbox{and} \quad
 1 \geq \mu_1 \geq \mu_2 \geq \ldots \geq \mu_p>0,
$
satisfying $\Sigma^\top \Sigma+\Omega^\top \Omega=I_{p}$.  The
{\em generalized singular values $\gamma_i$} of $(A,L)$ are defined
by the ratios $\gamma_i=\sigma_i/\mu_i$ ($i=[p]$).

In this section we design an algorithm that is aimed at the case when 
$n \gg d$. The general strategy
is to design a distribution on matrices of size $m$-by-$n$ ($m$ is a parameter), 
sample an $S$ from that distribution, and solve 
$\widetilde{x} \equiv \argmin_{x\in\R^d} \norm{S(Ax-b)}^2 +  \lambda\norm{Lx}^2\,$. 

The following lemma defines conditions on the distribution of $S$ that guarantees 
 Eq.~\eqref{eq:goal ridge regression} holds with constant probability
(which can be boosted to high probability by repetition and taking the minimum
objective).

\define{lem:reg}{Lemma}{
Let $x^*\in\R^d$, $A \in \R^{n\times d}$ and $b \in \R^n$ as above.
Let $U_1\in\R^{n\times d}$ denote the first $n$ rows of an orthogonal basis
for $\twomat{A}{\sqrt{\lambda} L } \in \R^{ (n+p) \times d}$.
Let sketching matrix $S\in\R^{m\times n}$ have a distribution such that with constant probability
\begin{align*}
\mathrm{(\RN{1})} ~
\norm{
U_1^\top S^\top S U_1 - U_1^\top U_1}\le 1/4,
\end{align*}
and
\begin{align*}
\mathrm{(\RN{2})} ~
\| U_1^\top (S^\top S -I) (b-Ax^*) \|_2 \le \sqrt{\eps\OPT/2}.
\end{align*}
Let $\widetilde{x}$ denote $\argmin_{x\in\R^d} \norm{S(Ax-b)}^2 +  \lambda\norm{Lx}^2$.
Then with probability at least $9/10$,
\begin{align*}
\norm{A\widetilde{x}-b}^2 + \lambda\norm{L\widetilde{x}}^2\le (1+\eps)\OPT.
\end{align*}
}
\state{lem:reg}

Define the \emph{statistical dimension}  for P-Splines as follows:
\begin{definition}[Statistical Dimension for P-Splines]
For S-Spline in Eq.~\eqref{eq ridge1}, the statistical dimension is defined as $\sd_\lambda(A,L)   =  \sum_i 1/(1+ \lambda/\gamma_i^2)+d-p$.
\end{definition}
The following theorem shows that  there is a sparse subspace embedding matrix $S \in \mathbb{R}^{m\times n}$ (e.g., {\sc CountSketch}), with $m \geq K(\sd_\lambda(A,L)/\epsilon + \sd_\lambda(A,L)^2)$,  that satisfies Property (\RN{1}) and (\RN{2}) of Lemma~\ref{lem:reg}, and hence achieves an $\epsilon-$approximation solution to problem~\ref{eq ridge1}:

\define{thm:size_of_S}{Theorem}{ \rm{(P-Spline regression) }
 There is a constant $K>0$ such that for  $m \ge K(\eps^{-1} \sd_\lambda(A,L) + \sd_\lambda(A,L)^2)$ and $S\in\R^{m\times n}$ a sparse embedding matrix (e.g., {\sc Countsketch}) with $SA$
  computable in $O(\nnz(A))$ time, Property (\RN{1}) and (\RN{2}) of Lemma~\ref{lem:reg} apply, and with constant probability
the corresponding $\widetilde{x} = \argmin_{x\in\R^d} \norm{S(Ax-b)} + \lambda\norm{Lx}^2$
is an $\eps$-approximate solution to $\min_{x\in\R^d} \norm{b-Ax}^2 + \lambda\norm{Lx}^2$.
}
\state{thm:size_of_S}
Note  $\sd_\lambda(A,L)$ is   upper bounded by $d$. The above theorem shows that the statistical dimension allows us to design smaller sketch matrices whose size only depends on $O(\poly(\sd_\lambda(A,L)/\epsilon))$ instead of $O(\poly(d/\epsilon))$, without sacrificing the approximation accuracy. 


\vspace{-3mm}
\subsection{Tensor Sketching for Multi-Dimensional P-Spline}
\begin{algorithm}[tb]
 \caption{P-Spline Tensor product regression}
  \label{alg:tensorregression_p}
  \begin{algorithmic}[1]
  \Procedure{PTRegression}{$A,b,K,L,\epsilon,\delta$}
    \State $m \leftarrow K(\epsilon^{-1}\sd_\lambda(\mathcal{A},L) +\sd_\lambda(\mathcal{A},L)^2)$
    \State Choose $S$ to be a $m \times n$ {\sc TensorSketch} matrix
    \State Compute $S(A_1\otimes A_2\otimes \cdots \otimes A_q)$ and $Sb$
    \State $ \widetilde { x} = \argmin_{x\in\R^d} \norm{S(Ax-b)}^2 +  \lambda\norm{Lx}^2$.
    \State \Return $ \widetilde{x}$
    \EndProcedure   
     \end{algorithmic}
\end{algorithm}
\vspace{-2mm}
Tensor products allow a natural extension of one-dimensional P-spline smoothing to multi-dimensional P-Spline. We focus on 2-dimensional P-Spline  but our results can be generalized to the multi-dimensional setting. Assume that in addition to $u$ we have a second explanatory variable $v$. We have data triples $(u_i, v_j , b_{{(i-1)\cdot n_2+j}} )$ for $i = 1, \ldots, n_1$ and $j=1,\ldots, n_2$.  We seek a smooth surface $f(u,v)$ which gives a good approximation to the response $b$. Let $A_1$, $n\times d_1$, be a B-spline basis along $u$, and $ A_2$, $n\times d_2$, be a B-spline basis along $v$. We form the tensor product basis as $A_1\otimes A_2$. When $n_1$ and $n_2$ are large, we do not  compute $A_1\otimes A_2$. We apply {\sc tensorsketch} here to avoid explicitly forming $A_1\otimes A_2$ to speed up computation. 


Let $X = [ x_{kl}]$ be a $d_1 \times  d_2$ matrix of coefficients. Then, for given $X$,  the value fit at $(u, v)$ is
$
	f(u,v)=\sum_k \sum_l A_{2,k}(v)A_{1,l}(u) x_{kl}
$
and so $X$ may be chosen using least squares by minimizing
$
	\sum_{i,j}[b_{{(i-1)\cdot n_2+j}} - f(u_i,v_i)]^2 =\sum_i\left[b_{{(i-1)\cdot n_2+j}} -\sum_k \sum_l A_{2,k}(v_i)A_{1,l}(u_i) x_{kl} \right]^2.
$
Using Kronecker product, the above minimization can be written in the  form $
	\min \left\|b-\cA x\right\|_2, 
	$
where $\cA=A_1 \otimes A_2  \in \R^{n_1 n_2 \times d_1 d_2 }
$  and $x=\textrm{ vec}(X)$. Again the P-spline approach minimizes the penalized least-squares function
$
	\left\|b- (A_1 \otimes A_2) x \right\|_2^2+\lambda \|L x\|_2^2.
$
 Consider the tensor p-spline regression problem
$
	\min_{x}\|(A_1\otimes A_2\otimes \cdots \otimes A_q)x-b\|_2^2+ \lambda\norm{Lx}^2,
$
where $L\in \R^{p\times n}$, $A_i\in \R^{n_i \times d_i},\, i=1,\ldots, q$ and $b\in \R^{n}$. Let $S \in \R^{m\times n}$ be the matrix form of {\sc TensorSketch} of Section \ref{sec:background}. Algorithm~\ref{alg:tensorregression_p} summarizes  the procedure for efficiently solving multi-dimensional P-Spline. 

Let $\cA=A_1\otimes A_2\otimes \cdots \otimes A_q$. Replacing the matrix $A$ in Theorem~\ref{thm:size_of_S} with $\mathcal{A}$, we have the following corollary for multi-dimensional P-Spline: 

\begin{corollary}[P-Spline tensor regression]\label{cor size of S}
Suppose $\lambda \leq \sigma_1^2/\epsilon$. 
 There is a constant $K>0$ such that for  $m \ge K(\eps^{-1} \sd_\lambda(\cA,L) + \sd_\lambda(\cA,L)^2)$ and $S\in\R^{m\times n}$ a {\sc TensorSketch}  matrix  with $S\cA$
  computable in $O(\nnz(\cA))$ time, Property (\RN{1}) and Property (\RN{2}) of Lemma~\ref{lem:reg} apply, and with constant probability
the corresponding $\widetilde{x} = \argmin_{x\in\R^d} \norm{S(\cA x-b)} + \lambda\norm{Lx}^2$
is an $\eps$-approximate solution to $\min_{x\in\R^d} \norm{b-\cA x}^2 + \lambda\norm{Lx}^2$.
\end{corollary}
\vspace{-3mm}
\section{TENSOR PRODUCT ABSOLUTE DEVIATION REGRESSION }\label{sec: ell_p}
\vspace{-3mm}
We extend our previous results for the $\ell_2$ norm (i.e., least squares regression) to general $\ell_p$ norms, with a focus on $p = 1$ (i.e., absolute deviation regression). Specifically we consider $\min_x \|\cA x-b\|_p$, where $\cA=A_1\otimes A_2\otimes \cdots \otimes A_q$. 
We will show in this section that with probability at least $2/3$, 
we can quickly find an $\tilde{x}$ for which
$
\|\cA \tilde{x}-b\|_1 \leq (1+\eps)\min_x \|\cA x-b\|_1.
$

As in \cite{cw13}, for each $i$, in $O(\nnz(A_i) \log (d_i)) + O(r_i^3)$ time we can replace the input matrix $A_i\in \R^{n_i \times d_i}$ with a new matrix with the same column space of $A_i$ and full column rank. We therefore assume $\cA$ has full rank in what follows. 

Suppose $S$ is the {\sc TensorSketch} matrix defined in Section \ref{sec:background}.  Let $w_i \in \mathbb{N}$ 
and assume $w_i \mid n_i$. 
Split $A_i$ into $n_i/w_i$ matrices $A_1^{(i)}, \ldots, A_{n_i/w_i}^{(i)}$,
each $w_i \times d_i$, so that $A_j^{(i)}$ is the submatrix of $A_i$ indexed by the $j$-th block of $w_i$ rows. Note $\cA$ can be written as:
\begin{align}
\begin{bmatrix}
	A_1^{(1)} \otimes A_1^{(2)}\otimes \cdots \otimes A_1^{(q-1)}\otimes A_1^{(q)} \\
	A_1^{(1)} \otimes A_1^{(2)}\otimes \cdots \otimes A_1^{(q-1)}\otimes  A_2^{(q)} \\
	\vdots \\
	A_{n_1/w_1}^{(1)} \otimes A_{n_2/w_2}^{(2)}\otimes \cdots \otimes A_{n_{q-1}/w_{q-1}}^{(q-1)}\otimes  A_{n_q/w_q}^{(q)} . \nonumber
\end{bmatrix}.
\end{align}
For each $A_{i_1}^{(1)} \otimes A_{i_2}^{(2)}\otimes \cdots \otimes A_{i_q}^{(q)} $, we can use the {\sc TensorSketch} matrix $S_{i_1i_2\ldots i_q} \in \R^{m \times \prod_{i=1}^q w_i}$, where we set $m \geq 100\prod_{i=1}^q d_i^2 (2+3^q)/\eps^2$, such that with probability at least $.99$, $\|S_{i_1i_2\ldots i_q}  A_{i_1}^{(1)} \otimes A_{i_2}^{(2)}\otimes \cdots \otimes A_{i_q}^{(q)}  x\|_2 = (1\pm\eps)\|A_{i_1}^{(1)} \otimes A_{i_2}^{(2)}\otimes \cdots \otimes A_{i_q}^{(q)} x\|_2$ simultaneously for all $x\in \R^{d}$ as $S_{i_1i_2\ldots i_q}$ is an oblivious subspace embedding (Lemma.~\ref{lem:ose}) . Now we use Algorithm 4 from \citep{liang2014improved} to boost the success probability by computing $t = O(\log(1/\delta))$ independent {\sc TensorSketch} products $S^{(j)}_{i_1i_2\ldots i_q}  A_{i_1}^{(1)} \otimes A_{i_2}^{(2)}\otimes \cdots \otimes A_{i_q}^{(q)}$, $j = [t]$, each with only
constant success probability, and then running a cross validation procedure similar to that in Algorithm 4 of \cite{liang2014improved}, to find one which succeeds with probability   at least $1-\delta  $:


\begin{lemma}[\citep{liang2014improved}]
\label{lem:boost}
	For $\delta, \epsilon \in (0,1)$, let $t=O(\log 1/\delta)$ and $m \geq 100\prod_{i=1}^q d_i^2(2+3^q)/\epsilon^2$. Running algorithm Algorithm 4 in \citep{liang2014improved} with parameters $t,m$,
	we can obtain a {\sc tensorsketch} $S \in \mathbb{R}^{m \times \prod_{i=1}^q w_i } $ such that with probability at least $1-\delta$, $\|S A_{i_1}^{(1)} \otimes A_{i_2}^{(2)}\otimes \cdots \otimes A_{i_q}^{(q)}  x\|_2 = (1\pm\eps)\|A_{i_1}^{(1)} \otimes A_{i_2}^{(2)}\otimes \cdots \otimes A_{i_q}^{(q)} x\|_2$ for all $x \in \mathbb{R}^d$.
	
\end{lemma}

After computing a {\sc tensorsketch} $S_{i_1i_2\ldots i_q} \in \mathbb{R}^{m \times \prod_{i=1}^q w_i}$ using lemma~\ref{lem:boost} for each row-block $A_{i_1}^{(1)}\otimes ...\otimes A_{i_q}^{(q)}$, we compose a single {\sc tensorsketch} for $\mathcal{A}$ as $S = {\rm{diag}}(S_{1,\cdots,1},\cdots,S_{i_1,\cdots,i_q},\cdots,S_{(n_1/w_1),\cdots,(n_q/w_q)} ) \in \mathbb{R}^{m \prod_{i=1}^q n_i/w_i \times \prod_{i=1}^q {n_i} }$, which is defined as:
\begin{align}
 \begin{bmatrix} 
     S_{1,\cdots,1} & &  &\\ 
    & \ddots & &  &\\ 
     & & S_{i_1,\cdots, i_q} & & \\ 
    & & &  \ddots     &\\
    & & & & S_{(n_1/w_1),\cdots,(n_q/w_q)} 
\end{bmatrix},\nonumber
\end{align} where each block on the diagonal is from Lemma~\ref{lem:boost}.
Note that $\cA$ has in total $\prod_{i=1}^q (n_i/w_i)$ many blocks. Using Lemma~\ref{lem:boost}  with a union bound over all  blocks of $\cA$,  we have the following theorem which shows $S$ is an oblivious subspace embedding for $\cA$ in $\ell_2$ norm:
\begin{theorem}[$\ell_2$ OSE for tensor matrices]\label{thm:jlmain}
 Given $\delta, \epsilon \in (0,1)$, let $S\in \mathbb{R}^{m \prod_{i=1}^q n_i/w_i \times \prod_{i=1}^q n_i }$ denote the matrix that has $\prod_{i=1}^q n_i /w_i$ diagonal block matrices where each diagonal block $S_{i_1i_2\ldots i_q} \in \mathbb{R}^{m \times \prod_{i=1}^q w_i }$ is from Lemma~\ref{lem:boost}. 
With probability at least $1-\prod_{i=1}^q (n_i/w_i) \delta$, $\|S\cA x\|_2 = (1 \pm \eps)\|\cA x\|_2,\forall x \in \mathbb{R}^d$.
\end{theorem}

It is known that for any matrix $A\in\mathbb{R}^{n\times r}$, we can compute a change of basis $U\in\mathbb{R}^{r\times r}$ such that $AU$ is an $(\alpha,\beta,p)$ well-conditioned basis of $A$ (see Definition~\ref{def:wcb}), in time polynomial with respect to $n,r$ \citep{ddhkm09}. Specifically, Theorem 5 in \citep{ddhkm09} shows that we can compute a change of basis $U$ for which $AU$ is a well-conditioned basis of $A$ in time $O(nr^5\log (n))$ time.
However we cannot afford to directly use the results from \citep{ddhkm09} to compute a well-conditioned basis for $\mathcal{A}$, which requires time at least $\Omega(n)$. Instead we compute a well-conditioned basis for $\mathcal{A}$ through a sketch $S\cA$ where $S$ is the {\sc tensorsketch} from Theorem~\ref{thm:jlmain}. 

Specifically we define the following procedure {\sf Condition}($\cA$) for computing a well-conditioned basis for $\cA$.  
Given  $\cA=A_1\otimes \cdots \otimes A_q  \in \mathbb{R}^{n\times r}$, 1) Compute $S \cA$; 2) Compute a $d \times d$ change of basis
matrix $U$ so that $S \cA U$ is an $(\alpha, \beta, p)$-well-conditioned basis of the
column space of  $S \cA$; 3) Output $\cA U/(d\gamma_p)$,
where $\gamma_p \equiv \sqrt{2} t^{1/p-1/2}$ for $p\le 2$,
and $\gamma_p \equiv \sqrt{2} w^{1/2-1/p}$ for $p\ge 2$, where $t= \prod_{i=1}^q m_i$ and $w=\prod_{i=1}^q w_i$. The following Lemma~\ref{lem:lpl2} (the proof can be found in the Appendix) is the analogue of that in \cite{CDMMMW} proved for the Fast Johnson Lindenstauss
Transform. 

\define{lem:lpl2}{Lemma}{
  For any $p\geq 1$. {\sf Condition}$(\cA)$ computes $\cA U/(d\gamma_p)$ which is an $(\alpha, \beta \sqrt{3} d (tw)^{|1/p-1/2|} , p)$-well-conditioned basis of $\cA$, with probability at least $1-\prod_{i=1}^q (n_i/w_i) \delta$.
  }
  \state{lem:lpl2}
Lemma~\ref{lem:lpl2} indicates that {\sf Condition$(\cA)$} computes a $(\alpha, \beta \cdot \poly(\max(d, \log n)) , p)$-well-conditioned basis. A well-conditioned basis can be used 
to solve $\ell_p$ regression problems, via sampling a subset of rows
of the well-conditioned basis $\mathcal{A}U$ with probabilities proportional to the $p$-th power of the $\ell_p$ norm of the rows \citep{w14}.  However the first issue for sampling is that  we cannot afford to compute $\mathcal{A}U$ as this requires $O(\nnz(\mathcal{A})d)$ time. To fix this, we apply a Gaussian sketch matrix $G \in \mathbb{R}^{d\times \log(n)}$ with i.i.d normal random variables \citep{dmmw11} on the right hand side of $\mathcal{A}U$. Note that $\mathcal{A}UG$ can be computed efficiently by first compuing $UG$ and then $\mathcal{A}(UG)$ in time $O(d^2\log (n) + nnz(\mathcal{A})\log(n))$.  The second issue is that even with $\mathcal{A}UG$, computing the $\ell_p$ norm of each row of $\mathcal{A}UG$ takes $O(n)$ time, but we want sublinear time.  This leads us to the following sampling technique. 

The high level idea is that since $\mathcal{A}UG$ only has $O(\log(n))$ columns, we can afford to sample columns of $\mathcal{A}UG$ with probability proportional to the $p$-th power of the $\ell_p$ norms of the columns, if we can efficiently estimate the $\ell_p$ norms of the columns (note that na\"ively computing the $\ell_p$ norm of a column also takes $O(n)$ time). Let us denote the first of column of $UG$ as $e\in \mathbb{R}^{\log(n)}$. We focus on how to efficiently estimate the $\ell_p$ norm of the first column of $\mathcal{A}UG$, which is $\mathcal{A}e$, and all the left columns can be estimated in the same way. We first reshape the vector $\mathcal{A}e$ into a 2-d matrix 
\begin{align}\label{eq:reshape_2d}
M=A_1\otimes ...\otimes A_{q_1} E (A_{q_1+1}\otimes...\otimes A_{q})^\top, \vspace{-2mm}
\end{align} 
where $M\in \mathbb{R}^{(n_1n_2...n_{q_1})\times (n_{q_1+1}...n_q})$, $E$ is obtained from reshaping $e$ into a $(d_1d_2...d_{q_1})\times (d_{q_1+1}..d_q)$ matrix and $q_1\in [1,q]$ is chosen such that  $(n_1n_2...n_{q_1}) \approx (n_{q_1+1}...n_q)$. Namely we reshape the column $\mathcal{A}e$ into a (nearly) square matrix. Focusing now on $p = 1$, note that $\|\mathcal{A}e\|_1 = \sum_{i=1}^{n_{q_1+1}...n_q} \|M_{*,i}\|_1$.  Hence to estimate $\|\mathcal{A}e\|_1$, we only need to estimate the $\ell_1$ norm of the columns of $M$.

Let us apply a sketch matrix ${R}\in \mathbb{R}^{O(\log(n))\times\prod_{i=1}^{q_1}n_i}$, whose entries are sampled i.i.d. from the Cauchy distribution, to the left hand side of $M$. For the $i$'th column of $M$, let us define random variables $z_l^{i} = R_{l,*}^\top M_i$, for $l\in [O(\log(n))]$, where $R_{l,*}$ is the $l$'th row of $R$. Due to the 1-stability property of the Cauchy distribution, we have that $\{z_{l}^i\}_{l=1}^{O(\log n)}$ are $O(\log n)$ independent Cauchys scaled by $\|M_i \|_1$. Applying a Chernoff bound to independent half-Cauchys (see  Claims 1, 2 and Lemmas 1, 2 in \cite{Indyk2006}), we have
$0.5\|M_i\|_1 \leq \texttt{median}_{l \in [O(\log n)]}\{ |z_j^{i}| \}\leq 1.5\|M_i\|_1$ with probability at least $1-2 e^{-cO(\log(n))}$, with constant $c\geq 0.07$.  Denote the median of $\{|z_l^i|\}_{l=1}^{O(\log (n))}$ as $\lambda_i$. By a union bound over all columns of $M$, we have that with probability at least $1- n_{[q] \backslash [q_1]}  2e^{-cO(\log(n))}$:
\begin{align}
\label{eq:lambda_i}
& \lambda_i =  (1\pm 0.5)\|M_i\|_1, \forall i \in [ n_{[q] \backslash [q_1]} ], \\
\label{eq:lambda}
&  \lambda_e = \sum_{i=1}^{n_{[q] \backslash [q_1]} }\lambda_i = (1 \pm 0.5)\|\mathcal{A}e\|_1
\end{align}
where $n_{[q] \backslash [q_1]} = \prod_{i=q_1+1}^q n_i$.

\begin{algorithm}[tb]
 \caption{$\ell_{1}$ tensor product regression}
  \label{alg:l1tensorregression}
  \begin{algorithmic}[1]
  \Procedure{\textsc{L1TRregression}}{$A,b,\epsilon,\delta$}
    \State Construct a {\sc tensorsketch} $S \in \mathbb{R}^{(m\prod_{i=1}^q \frac{n_i}{w_i}) \times n}$. 
    
    \State Run {\sf Condition}($\cA$) using $S$ to compute $U/(d\gamma_p)$.
    
    \State Generate a Gaussian  matrix $G\in \mathbb{R}^{d\times O(\log(n))}$ and a Cauchy sketch matrix $R\in\mathbb{R}^{\log(n)\times n}$.
    \For{for each column $e$ in $UG$}
    	\State Reshape $\cA e$ to $M$ (Eq.~\ref{eq:reshape_2d}).
	 \State Compute $\lambda_i$ and $\lambda_e =\sum_i\lambda_i$ (Eq~\ref{eq:lambda_i} and \ref{eq:lambda}).
    \EndFor
    \For {$i \in [\sqrt{\prod_{i=1}^q w_i}\poly(d)$]}
    	\State Sample a column $(AUG)_{*,e}$ with probability proportional to $ \lambda_e$. \label{line:sample_AUG_c}.   \Comment{$ e  \in {[O(\log(n))]}$} \label{line:sample_start}
	\State Reshape $(AUG)_{*,e}$ to $M$, sample a column $M_{*,j}$ with probability proportional to $\lambda_j$.
	\State Sample an entry $M_{k,j}$ with probability proportional to $|M_{k,j}|$.
	\State Convert $(k, j)$ back to the corresponding row index in $\cA$, denoted as $r_i$. \label{line:sample_end}
    \EndFor




    \State  
 $
	\widetilde{x} \leftarrow \min_{x}\|D(A_1\otimes A_2\otimes \cdots \otimes A_q)x-Db\|_1,
$
where $D$ is a diagonal matrix to select the $r_1, r_2, ..., r_N$-th rows of $\cA $ and $b$ with $N = \sqrt{\prod_{i=1}^q w_i\poly(d)}$.
\State \Return $\widetilde{x}$ \Comment{$ \widetilde { x}  \in \R^{d_1d_2\cdots d_q}$}
\EndProcedure
    \end{algorithmic}
\end{algorithm}

Note that since $\mathcal{A}UG$ only has $O(\log n)$ columns, we can afford to compute the $\ell_1$ norm of all the columns using the above procedure. Let us denote the $\ell_1$ norms of the columns of $AUG$ by $\lambda_1, \lambda_2, ..., \lambda_{O(\log n)}$. We can sample a column $(\mathcal{A}UG)_{*,i}$ with probability proportional to $\lambda_i$  (Line~\ref{line:sample_AUG_c} in Alg.~\ref{alg:l1tensorregression}). Once we sample a column $\mathcal{A}UG_{*,i}$, we need to sample an entry $j\in [n]$ with probability proportional to the absolute value of the entry $|(\mathcal{A}UG)_{j,i}|$. As we cannot afford to compute $|(\mathcal{A}UG)_{j,i}|$ for all $j\in [n]$, we use the reshaped 2-d matrix $M$ of $(\cA UG)_{*,i}$. Note that sampling an entry in $M$ with probability proportional to the absolute values of entries of $M$ is equivalent to sampling an entry  $j\in [n]$ from $(\cA UG)_{*,i}$ with probability proportional to the absolute value of the entries of $(\cA UG)_{*,i}$.  We first sample a column $M_{*,j}$ from all the columns of $M$ with probability proportional to $\lambda_j$ for $j \in [\prod_{i=q_1+1}^q n_i]$. We then sample an entry $M_{k,j}$ with probability proportional to $|M_{k,j}|$ for $k\in [\prod_{i=1}^{q_1} n_i]$. Noting that $k\in [\prod_{i=1}^{q_1} n_i]$ and $j\in [\prod_{i=q_1+1}^q n_i]$, the pair $(k,j)$ uniquely determines a corresponding row index $r$ in $\cA$, for some $r\in [\prod_{i=1}^q n_i]$. Hence we successfully sample a row from $\cA UG$ without ever  computing the $\ell_p$ norm of the rows. The above sampling procedure is summarized in Line~\ref{line:sample_start} to Line~\ref{line:sample_end} in Alg.~\ref{alg:l1tensorregression}. We use the above procedure to sample $\sqrt{\prod_{i=1}^q w_i}\poly(d)$ rows of $\mathcal{A}$. Let $D$ be a diagonal matrix that selects the corresponding sampled rows from $\mathcal{A}$. We can now solve a smaller ADL problem as $\min_{x}\|D\mathcal{A} - Db\|_1$.  Note that our analysis focuses on the $\ell_1$ norm. We can extend the analysis to general $\ell_p$ norms by using a sketching matrix $R\in\mathbb{R}^{O(\log n)\times \prod_{i=1}^{q_1}n_i}$ with entries sampled i.i.d. from a $p$-stable distribution for $p\in [1,2]$.  

We now present our main theorem and defer the proofs to the appendix:
\define{thm:lp_running}{Theorem}{ \rm{(Main result)}
Given $\epsilon\in(0,1)$, 
$\cA\in\R^{n \times d}$ and $b\in\R^{n}$, Alg.~\ref{alg:l1tensorregression} computes $\widehat{x}$ such that 
with probability at least $1/2$, 
$\|\cA \hat{x}-b\|_1\le (1+\epsilon) \min_{x\in \R^d}\|\cA x-b\|_1$.
For the special case when $q = 2$, $n_1=n_2$, the algorithm's running time is $O({n_1}^{3/2} \poly(\prod_{i=1}^2 d_i /\epsilon) )$. 
}
\state{thm:lp_running}
For the special case where $q = 2$ and $n_1=n_2 $, we can see from theorem~\ref{thm:lp_running} our algorithm computes $\widehat{x}$ in time $O({n_1}^{3/2}\poly(\hat{d}))$, which is faster than $O({n_1}^2)$---the time needed for  forming $A_1\otimes A_2$.  
Note that we can run Alg.~\ref{alg:l1tensorregression} $O(\log(1/\delta))$ times independently and pick the best solution among these independent runs to boost the success probability to be $1-\delta$, for $\delta\in[0,1)$.

\vspace{-3mm}
\section{NUMERICAL EXPERIMENTS}
\vspace{-2mm}
We generate matrices $A_1$, $A_2$ and $b$ with all entries sampled i.i.d from a normal distribution. The baseline we compared to is directly solving regression without sketching.  We let ${\tt T_1}$ be the time for directly solving the regression problem, and ${\tt T_2}$ be the time of our algorithm. The time ratio is $r_t={\tt T_2}/ {\tt T_1}$. The relative residual percentage is defined by $
r_e=\frac{100\left|\,  \left\|(A_1 \otimes A_2)\widetilde { x}-b\right\|_2-\left\|(A_1 \otimes A_2)x^*-b\right\|_p \, \right|}{\left\|(A_1 \otimes A_2)x^*-b\right\|_p},
$
where $\widetilde { x}$ is the output of our algorithms and $x^*$ is the optimal solution. Throughout the simulations, we use a moderate input matrix size in order  to accommodate the brute force algorithm and to compare to the exact solution.
\begin{example} [$\ell_2$ Regression]\label{ex:1}
We create a design matrix with moderate size by fixing $n_1=n_2=300$ and $d_1=d_2=15$. Thus $\cA \in \R^{90000\times 225}$.   
We do 10 rounds to compute the mean values of $r_e$ and $r_t$, which is reported in Table \ref{ta:1} (Left).

\vspace{-5pt}
\end{example}

        From Table \ref{ta:1} (left), we can see that when the number $m$ of sampled rows in Algorithm  \ref{alg:tensorregression}
        increases from 8000 to 12000, the mean values of $r_e$ decrease while the mean values of $r_t$ increase. In general we can see that we can achieve around $1\%$ relative error while being 5 times faster than the direct method.

\begin{table}
\caption{Examples \ref{ex:1} and \ref{ex:2}: the values of $r_e$ and $r_t$ with respect to different sampling parameters $m$.}
\label{ta:1}
\begin{center}
\resizebox{1.0 \columnwidth}{!}{
	\begin{tabular}{|c|c|c|c||c|c|c|c|}
\hline
	 & $m$ & $r_e$ &  $r_t$ & & $m$ & $r_e$ & $r_t$\\
	\hline 
 \multirow{3}{1em}{$\ell_2$} & 8000 & 1.79\% & 0.11 & \multirow{3}{1em}{$\ell_1$} & 8000  & 1.89\%  & 0.06\\
  & 12000 & 1.24\% & 0.18 & & 12000  & 1.33\%  & 0.11 \\
  & 16000 & 1.01\% & 0.25 & & 16000  & 0.992\% & 0.18\\
	\hline
\end{tabular}
}
\end{center}
\vspace{-10pt}
\end{table}

%

\begin{example} [$\ell_1$ Regression]
\label{ex:2}
We set $n_1 = n_2 = 300$, $d_1 = d_2 = 15$. For  $\min_x \|Ax-b\|_1$, we solve it by a Linear Programming solver in Gurobi \citep{gurobi}.  We tested different numbers of sampled rows $m$ (the number of rows in $D$ in Alg.~\ref{alg:l1tensorregression}).  The results are summarized in Table~\ref{ta:1} (Right).
\vspace{-5pt}
\end{example} 
As we can see from  Table~\ref{ta:1} (right), our method is around 10 times faster than directly solving the problem, with relative error only around $1\%$.

\begin{example} [P-Spline Regression]
\label{ex:3}
For P-spline regression, $L_3$ is fixed. We use 30 knots and cubic B-splines. The data $u=(u_i)\in \R^{n_1}, v=(v_j)\in \R^{n_2}$ and $b\in \R^{n_1 n_2}$ are generated i.i.d from the normal distribution. We compute the B-spline basis matrices $A_1 \in \R^{n_1 \times d_1}$ and $A_2 \in \R^{n_2 \times d_2}$ separately,  if there are $d_1$ and $d_2$ basis functions for the B-spline. We set $n_1=n_2$, $d_1=d_2$, with $n_1n_2 = 10^4$ and $d_1d_2 = 529$. The original and sketched P-spline regression problem are both solved by {\sc Regularization Tools} \citep{hansen1994} via computing their GSVDs. We test different choices of  $\lambda$. 
The results are shown  in Table \ref{ta:2}. 
\vspace{-5pt}
\end{example}

From Table \ref{ta:2}, sampling only 20\% of the rows can give around $0.05\%$ relative error, while the computation time is half of the time of directly solving the p-spline. 

\begin{table}
\caption{Example \ref{ex:3}: the mean values of $r_e$ and $r_t$ with respect to different sampling parameters $m$ and regularization parameters $\lambda$. }
\label{ta:2}
\begin{center}
\resizebox{0.6\columnwidth}{!}{
	\begin{tabular}{|c|c|c|c|c|}
\hline
$\lambda $&	$m$ & $r_e$ &  $r_t$\\
	\hline 
 1& 2000 & 4.43e-2\% & 0.52\\
& 4000 & 2.99e-2\% & 0.70 \\
& 6000 & 7.92e-2\% & 0.91\\
\hline 
0.1& 2000 & 6.98e-2\% & 0.47 \\
& 4000 & 4.07e-2\% & 0.72 \\
& 6000 & 5.41e-2\% & 0.94\\
\hline 
0.01& 2000 & 3.78e-2\% & 0.46 \\
& 4000 & 1.07e-1\% & 0.71 \\
& 6000 & 2.97e-2\% & 0.95\\
	\hline
\end{tabular}
}
\end{center}
\vspace{-10pt}
\end{table}

\vspace{-2mm}
\section{CONCLUSION}
\vspace{-2mm}
We propose algorithms for efficiently solving tensor product least squares regression, regularized P-splines, as well as tensor product least absolute deviation regression, using sketching techniques. Our main contributions are: (1) we apply {\sc tensorsketch} to least square regression problems, (2) we propose new statistical dimension measures for P-splines, extending the previous statistical dimension defined only for classic Ridge regression, and (3) we extend {\sc tensorsketch} to $\ell_p$ norms and propose an algorithm that can solve tensor product $\ell_1$ regression in time sublinear in the time for explicitly computing the tensor product. Simulation results support our theorems and demonstrate that our algorithms are much faster than brute-force algorithms and can achieve approximate solutions that are close to optimal.

\section*{ACKNOWLEDGEMENT}
Huaian Diao is supported in part by the Fundamental Research Funds for the Central Universities under the grant 2412017FZ007. Wen Sun is supported in part by Office of Naval Research contract N000141512365.

\bibliographystyle{abbrvnat}
\bibliography{ref}

\begin{thebibliography}{35}
\providecommand{\natexlab}[1]{#1}
\providecommand{\url}[1]{\texttt{#1}}
\expandafter\ifx\csname urlstyle\endcsname\relax
  \providecommand{\doi}[1]{doi: #1}\else
  \providecommand{\doi}{doi: \begingroup \urlstyle{rm}\Url}\fi

\bibitem[Avron et~al.(2014)Avron, Nguyen, and Woodruff]{anw14}
H.~Avron, H.~Nguyen, and D.~Woodruff.
\newblock Subspace embeddings for the polynomial kernel.
\newblock In \emph{Advances in Neural Information Processing Systems(NIPS)},
  pages 2258--2266, 2014.

\bibitem[Avron et~al.(2016)Avron, Clarkson, and Woodruff]{acw16a}
H.~Avron, K.~L. Clarkson, and D.~P. Woodruff.
\newblock Sharper bounds for regression and low-rank approximation with
  regularization.
\newblock \emph{CoRR}, abs/1611.03225, 2016.

\bibitem[Carter and Wegman(1979)]{cw79}
L.~Carter and M.~N. Wegman.
\newblock Universal classes of hash functions.
\newblock \emph{J. Comput. Syst. Sci.}, 18\penalty0 (2):\penalty0 143--154,
  1979.

\bibitem[Charikar et~al.(2004)Charikar, Chen, and Farach-Colton]{ccf04}
M.~Charikar, K.~Chen, and M.~Farach-Colton.
\newblock Finding frequent items in data streams.
\newblock \emph{Theor. Comput. Sci.}, 312\penalty0 (1):\penalty0 3--15, 2004.

\bibitem[Chen and Plemmons(2010)]{chen2012}
D.~Chen and R.~J. Plemmons.
\newblock Nonnegativity constraints in numerical analysis.
\newblock In \emph{The birth of numerical analysis}, pages 109--139. World Sci.
  Publ., Hackensack, NJ, 2010.

\bibitem[Clarkson et~al.(2013)Clarkson, Drineas, Magdon-Ismail, Mahoney, Meng,
  and Woodruff]{CDMMMW}
K.~Clarkson, P.~Drineas, M.~Magdon-Ismail, M.~Mahoney, X.~Meng, and D.~P.
  Woodruff.
\newblock The fast {C}auchy transform and faster robust linear regression.
\newblock In \emph{SODA}, 2013.

\bibitem[Clarkson(2005)]{c05}
K.~L. Clarkson.
\newblock Subgradient and sampling algorithms for $\ell_1$ regression.
\newblock In \emph{Proceedings of the sixteenth annual ACM-SIAM symposium on
  Discrete algorithms (SODA)}, pages 257--266, 2005.

\bibitem[Clarkson and Woodruff(2013)]{cw13}
K.~L. Clarkson and D.~P. Woodruff.
\newblock Low rank approximation and regression in input sparsity time.
\newblock In \emph{Symposium on Theory of Computing Conference, STOC'13, Palo
  Alto, CA, USA, June 1-4, 2013}, pages 81--90.
  \url{https://arxiv.org/pdf/1207.6365}, 2013.

\bibitem[Dasgupta et~al.(2009)Dasgupta, Drineas, Harb, Kumar, and
  Mahoney]{ddhkm09}
A.~Dasgupta, P.~Drineas, B.~Harb, R.~Kumar, and M.~W. Mahoney.
\newblock Sampling algorithms and coresets for $\ell_p$ regression.
\newblock \emph{SIAM J. Comput.}, 38\penalty0 (5):\penalty0 2060--2078, 2009.

\bibitem[Drineas et~al.(2011)Drineas, Magdon-Ismail, Mahoney, and
  Woodruff]{dmmw11}
P.~Drineas, M.~Magdon-Ismail, M.~W. Mahoney, and D.~P. Woodruff.
\newblock Fast approximation of matrix coherence and statistical leverage.
\newblock \emph{CoRR}, abs/1109.3843, 2011.

\bibitem[Eilers and Marx(2006)]{eilers2006multidimensional}
P.~H. Eilers and B.~D. Marx.
\newblock Multidimensional density smoothing with p-splines.
\newblock In \emph{Proceedings of the 21st international workshop on
  statistical modelling}, 2006.

\bibitem[Eilers and Marx(1996)]{eilers96}
P.~H.~C. Eilers and B.~D. Marx.
\newblock Flexible smoothing with {$B$}-splines and penalties.
\newblock \emph{Statist. Sci.}, 11\penalty0 (2):\penalty0 89--121, 1996.

\bibitem[Eilers et~al.(2015)Eilers, Marx, and Durb\'an]{eilers15}
P.~H.~C. Eilers, B.~D. Marx, and M.~Durb\'an.
\newblock Twenty years of {P}-splines.
\newblock \emph{SORT}, 39\penalty0 (2):\penalty0 149--186, 2015.
\newblock ISSN 1696-2281.

\bibitem[Fausett and Fulton(1994)]{Fausett1994}
D.~W. Fausett and C.~T. Fulton.
\newblock Large least squares problems involving {K}ronecker products.
\newblock \emph{SIAM J. Matrix Anal. Appl.}, 15\penalty0 (1):\penalty0
  219--227, 1994.

\bibitem[Golub and Van~Loan(2013)]{golubVanLoan2013Book}
G.~H. Golub and C.~F. Van~Loan.
\newblock \emph{Matrix computations}.
\newblock Johns Hopkins Studies in the Mathematical Sciences. Johns Hopkins
  University Press, Baltimore, MD, 2013.

\bibitem[Gurobi~Optimization(2016)]{gurobi}
I.~Gurobi~Optimization.
\newblock Gurobi optimizer reference manual, 2016.
\newblock URL \url{http://www.gurobi.com}.

\bibitem[Hansen(1994)]{hansen1994}
P.~C. Hansen.
\newblock Regularization tools: A matlab package for analysis and solution of
  discrete ill-posed problems.
\newblock \emph{Numerical Algorithms}, 6\penalty0 (1):\penalty0 1--35, 1994.

\bibitem[Indyk(2006)]{Indyk2006}
P.~Indyk.
\newblock Stable distributions, pseudorandom generators, embeddings, and data
  stream computation.
\newblock \emph{J. ACM}, 53\penalty0 (3):\penalty0 307--323, 2006.

\bibitem[Kannan and Vempala(2009)]{kv09}
R.~Kannan and S.~Vempala.
\newblock Spectral algorithms.
\newblock \emph{Foundations and Trends in Theoretical Computer Science},
  4\penalty0 (3-4):\penalty0 157--288, 2009.

\bibitem[Liang et~al.(2014)Liang, Balcan, Kanchanapally, and
  Woodruff]{liang2014improved}
Y.~Liang, M.-F.~F. Balcan, V.~Kanchanapally, and D.~Woodruff.
\newblock Improved distributed principal component analysis.
\newblock In \emph{Advances in Neural Information Processing Systems}, pages
  3113--3121, 2014.

\bibitem[Mahoney(2011)]{m11}
M.~W. Mahoney.
\newblock Randomized algorithms for matrices and data.
\newblock \emph{Foundations and Trends in Machine Learning}, 3\penalty0
  (2):\penalty0 123--224, 2011.

\bibitem[Meng and Mahoney(2013)]{mm13}
X.~Meng and M.~W. Mahoney.
\newblock Low-distortion subspace embeddings in input-sparsity time and
  applications to robust linear regression.
\newblock In \emph{Proceedings of the forty-fifth annual ACM symposium on
  Theory of computing}, pages 91--100. ACM,
  \url{https://arxiv.org/pdf/1210.3135}, 2013.

\bibitem[Nelson and Nguy{\^e}n(2013)]{nn13}
J.~Nelson and H.~L. Nguy{\^e}n.
\newblock Osnap: Faster numerical linear algebra algorithms via sparser
  subspace embeddings.
\newblock In \emph{2013 IEEE 54th Annual Symposium on Foundations of Computer
  Science (FOCS)}, pages 117--126. IEEE, \url{https://arxiv.org/pdf/1211.1002},
  2013.

\bibitem[Oh and Yun(2005)]{oh2005}
S.~Oh, S.~Kwon and J.~Yun.
\newblock A method for structured linear total least norm on blind
  deconvolution problem.
\newblock \emph{Applied Mathematics and Computing}, 19:\penalty0 151--164,
  2005.

\bibitem[Pagh(2013)]{p13}
R.~Pagh.
\newblock Compressed matrix multiplication.
\newblock \emph{ACM Trans. Comput. Theory}, 5\penalty0 (3):\penalty0 9:1--9:17,
  2013.

\bibitem[Patrascu and Thorup(2012)]{pt12}
M.~Patrascu and M.~Thorup.
\newblock The power of simple tabulation hashing.
\newblock \emph{J. ACM}, 59\penalty0 (3):\penalty0 14, 2012.

\bibitem[Pham and Pagh(2013)]{pp13}
N.~Pham and R.~Pagh.
\newblock Fast and scalable polynomial kernels via explicit feature maps.
\newblock In \emph{Proceedings of the 19th ACM SIGKDD international conference
  on Knowledge discovery and data mining(KDD)}, pages 239--247. ACM, 2013.

\bibitem[Rousseeuw and Leroy(2005)]{rousseeuw2005robust}
P.~J. Rousseeuw and A.~M. Leroy.
\newblock \emph{Robust regression and outlier detection}, volume 589.
\newblock John wiley \& sons, 2005.

\bibitem[Sohler and Woodruff(2011)]{sw11}
C.~Sohler and D.~P. Woodruff.
\newblock Subspace embeddings for the $\ell_1$-norm with applications.
\newblock In \emph{Proceedings of the forty-third annual ACM symposium on
  Theory of computing (STOC)}, pages 755--764. ACM, 2011.

\bibitem[Song et~al.(2017{\natexlab{a}})Song, Woodruff, and Zhong]{swz17}
Z.~Song, D.~P. Woodruff, and P.~Zhong.
\newblock Low rank approximation with entrywise $\ell_1$-norm error.
\newblock In \emph{Proceedings of the 49th Annual Symposium on the Theory of
  Computing (STOC)}. ACM, \url{https://arxiv.org/pdf/1611.00898},
  2017{\natexlab{a}}.

\bibitem[Song et~al.(2017{\natexlab{b}})Song, Woodruff, and Zhong]{swz17b}
Z.~Song, D.~P. Woodruff, and P.~Zhong.
\newblock Relative error tensor low rank approximation.
\newblock \emph{arXiv preprint arXiv:1704.08246}, 2017{\natexlab{b}}.

\bibitem[Van~Loan(1992)]{vanbook}
C.~Van~Loan.
\newblock \emph{Computational frameworks for the fast {F}ourier transform},
  volume~10 of \emph{Frontiers in Applied Mathematics}.
\newblock Society for Industrial and Applied Mathematics (SIAM), Philadelphia,
  PA, 1992.

\bibitem[Van~Loan and Pitsianis(1993)]{vanLoanKron}
C.~F. Van~Loan and N.~Pitsianis.
\newblock Approximation with {K}ronecker products.
\newblock In \emph{Linear algebra for large scale and real-time applications
  ({L}euven, 1992)}, volume 232 of \emph{NATO Adv. Sci. Inst. Ser. E Appl.
  Sci.}, pages 293--314. Kluwer Acad. Publ., Dordrecht, 1993.

\bibitem[Woodruff(2014)]{w14}
D.~P. Woodruff.
\newblock Sketching as a tool for numerical linear algebra.
\newblock \emph{Foundations and Trends in Theoretical Computer Science},
  10\penalty0 (1-2):\penalty0 1--157, 2014.

\bibitem[Zheng et~al.(2012)Zheng, Liu, Sugimoto, Yan, and
  Okutomi]{zheng2012practical}
Y.~Zheng, G.~Liu, S.~Sugimoto, S.~Yan, and M.~Okutomi.
\newblock Practical low-rank matrix approximation under robust l 1-norm.
\newblock In \emph{Computer Vision and Pattern Recognition (CVPR), 2012 IEEE
  Conference on}, pages 1410--1417. IEEE, 2012.

\end{thebibliography}


\newpage
\onecolumn
\appendix 

\section{Background: {\sc CountSketch} and {\sc TensorSketch}}
We start by describing the {\sc CountSketch} transform~\cite{ccf04}.
Let $m$ be the target dimension.
When applied to $n$-dimensional vectors,
the transform is specified by a $2$-wise independent hash function $h:[n] \rightarrow [m]$
and a $2$-wise independent sign function $s:[n] \rightarrow \{-1,+1\}$. When applied to $v$, 
the value at coordinate
$i$ of the output,  $i = 1, 2, \ldots, m$ is $\sum_{j \mid h(j) = i} s(j) v_j$. Note that {\sc CountSketch}
can be represented as an $m \times n$ matrix in which the $j$-th column contains a single non-zero entry $s(j)$
in the $h(j)$-th row.

We now describe the {\sc TensorSketch} transform~\cite{p13}.
Suppose we are given  points $v_i \in \mathbb{R}^{n_i}$, where $i=1,\ldots, q$ and so $\phi(v_1,\ldots,v_q)=v_1\otimes v_2 \otimes \cdots \otimes v_q\in \mathbb{R}^{n_1n_2\cdots n_q}$,
and the target dimension is again $m$. The transform is specified using $q$ $3$-wise independent hash
functions $h_i: [n_i] \rightarrow [m]$,
and $q$ $4$-wise independent sign functions $s_i: [n_i] \rightarrow \{+1, -1\}$, where $i=1,\ldots, q$.
{\sc TensorSketch} applied to $v_1,\ldots \otimes v_q$ is then {\sc CountSketch} applied to $\phi(v_1,\ldots,v_q)$ with hash function
$H:[n_1n_2\cdots n_q] \rightarrow [m]$ and sign function $S:[n_1n_2\cdots n_q] \rightarrow \{+1, -1\}$ defined as follows:
$$
H(i_1, \ldots, i_q) = h_1(i_1) + h_2(i_2) + \cdots + h_q(i_q) \bmod m,
$$
and
$$
S(i_1, \ldots, i_q) = s_1(i_1) \cdot s_2(i_2) \cdots s_q(i_q),
$$
where $i_j\in [n_j]$. It is well-known that if $H$ is constructed this way, then it is $3$-wise independent~\cite{cw79,pt12}.
Unlike the work of Pham and Pagh~\cite{pp13}, which only used that $H$ was $2$-wise independent, our analysis needs this stronger property of $H$.

The {\sc TensorSketch} transform can be applied to $v_1,\ldots, v_q$ without computing $\phi(v_1,\ldots,v_q)$ as follows. Let $v_j=(v_{j_\ell}) \in \mathbb{R}^{n_j}$. First, compute the polynomials
$$
p_{\ell}(x) = \sum_{i=0}^{B-1} x^i \sum_{j_\ell \mid h_{\ell} (j_\ell) = i} v_{j_\ell} \cdot s_{\ell}(j_\ell),
$$
for $\ell = 1, 2, \ldots, q$. A calculation \cite{p13} shows
\begin{equation*}
\prod_{\ell = 1}^q p_{\ell}(x) \bmod (x^B -1) = \sum_{i=0}^{B-1} x^i  \sum_{(j_1, \ldots, j_q) \mid H(j_1, \ldots, j_q) = i} v_{j_1} \cdots v_{j_q} S(j_1, \ldots, j_q),
\end{equation*}
that is, the coefficients of the product of the $q$ polynomials $\bmod \ (x^m-1)$ form the value of {\sc TensorSketch($v_1,\ldots,v_q$)}. Pagh observed that this product of polynomials can be computed in
$O(q m \log m)$ time using the Fast Fourier Transform. As it takes $O(q \max(\nnz(v_i)))$ time
to form the $q$ polynomials, the overall time to compute {\sc TensorSketch}$(v)$ is
$O(q(\max(\nnz(v_i) )+ m \log m))$.

\section{{\sc TensorSketch} is an Oblivious Subspace Embedding (OSE)}\label{Sec:TensorSketch}

Let $S$ be the $ m \times (n_1 n_2\cdots n_q)$ matrix such that {\sc TensorSketch} $(v_1,\ldots,v_q)$ is $S \cdot \phi(v_1,\ldots,v_q)$ for a randomly
selected {\sc TensorSketch}. Notice that $S$ is a random matrix.
In the rest of the paper, we refer to such a matrix as a {\sc TensorSketch} matrix with an appropriate number of rows, i.e., the number of hash buckets.
We will show that $S$ is an oblivious subspace embedding for subspaces in $\R^{n_1n_2\cdots n_q}$ for appropriate values of $m$.
Notice that $S$ has exactly one non-zero entry per column. The index of the non-zero in the column
$(i_1, \ldots, i_q)$ is $H(i_1, \ldots, i_q) = \sum_{j=1}^q h_j(i_j)\bmod{m}$. Let $\delta_{a,b}$ be the
indicator random variable of whether $S_{a,b}$ is non-zero. The sign of the non-zero entry in column
$(i_1, \ldots, i_q)$ is $S(i_1,\ldots, i_q)=\prod_{j=1}^q s_j(i_j)$. We show that the embedding matrix $S$ of {\sc TensorSketch} can be used to approximate matrix product and is an oblivious subspace embedding (OSE). 

\begin{theorem}
Let $S$ be the $ m \times (n_1 n_2\cdots n_q)$ matrix such that  $$
\mbox{ {\sc TensorSketch}}(v_1,\ldots,v_q)
$$ is $S \cdot \phi(v_1,\ldots,v_q)$ for a randomly
selected {\sc TensorSketch}. The matrix $S$ satisfies the following two properties.
\begin{enumerate}
\item ($\mathrm{Approximate~Matrix~Product:}$) Let $A$ and $B$ be matrices with $n_1n_2\cdots n_q$ rows. For $m \geq (2+3^q)/(\eps^2\delta)$, we have
$$
\Pr_{S} \left[ \|A^\top S^\top  S B - A^\top  B\|_F^2 \le \eps^2 \|A\|_F^2 \|B\|_F^2\right] \geq 1-\delta.
$$
\item ($\mathrm{Subspace~Embedding:}$) Consider a fixed $k$-dimensional subspace $V$. If $m \geq k^2 (2+3^q)/(\eps^2\delta)$, then with probability at least $1-\delta$, $\|Sx\| = (1\pm\eps)\|x\|$ simultaneously for all $x\in V$.
\end{enumerate}
\end{theorem}

We establish the theorem via two lemmas as in \cite{acw16a}. The first lemma proves the approximate matrix product property via a careful second moment analysis.

\begin{lemma}[Approximate matrix product]\label{lem:mp} Let $A$ and $B$ be matrices with $n_1 n_2\cdots n_q$ rows. For $m \geq (2+3^q)/(\eps^2\delta)$, we have
\begin{align*}
\Pr_{S} \left[ \|A^\top S^\top  S B - A^\top  B\|_F^2 \le \eps^2 \|A\|_F^2 \|B\|_F^2 \right] \ge 1-\delta.
\end{align*}
\end{lemma}

\begin{proof}
The proof follows that in \cite{acw16a}. Let $C = A^\top  S^\top  S B$. We have
\begin{equation*}
C_{u,u'} = \sum_{t=1}^m \sum_{i,j\in {[n_1n_2\cdots n_q]}} S(i)S(j)\delta_{t,i}\delta_{t,j} A_{i,u} B_{j,u'}=\sum_{t=1}^m \sum_{i\ne j\in {[n_1n_2\cdots n_q]}} S(i)S(j)\delta_{t,i}\delta_{t,i} A_{i,u} B_{j,u'} + (A^\top  B)_{u,u'}
\end{equation*}
Thus, $\E[C_{u,u'}] = (A^\top B)_{u,u'}$.

Next, we analyze $\E[((C-A^\top B)_{u,u'})^2]$. We have
\begin{align*}
((C - A^\top B)_{u,u'})^2 =
\sum_{t_1, t_2 = 1}^m \sum_{i_1\ne j_1, i_2\ne j_2 \in {[n_1n_2\cdots n_q]}} &S(i_1)S(i_2)S(j_1)S(j_2) \cdot
 \delta_{t_1,i_1}\delta_{t_1,j_1}\delta_{t_2,i_2}\delta_{t_2,j_2} \cdot A_{i_1,u}A_{i_2, u}B_{j_1, u'}B_{j_2, u'}
\end{align*}
For a term in the summation on the right hand side to have a non-zero expectation, it must be the case that $\E[S(i_1)S(i_2)S(j_1)S(j_2)]\ne 0$. Note that $S(i_1)S(i_2)S(j_1)S(j_2)$ is a product of random signs (possibly with multiplicities) where the random signs in different coordinates in $\{1,\ldots, q\}$ are independent and they are 4-wise independent within each coordinate. Thus, $\E[S(i_1)S(i_2)S(j_1)S(j_2)]$ is either $1$ or $0$. For the expectation to be $1$, all random signs must appear with even multiplicities. In other words, in each of the $q$ coordinates, the 4 coordinates of $i_1, i_2, j_1, j_2$ must be the same number appearing 4 times or 2 distinct numbers, each appearing twice. All the subsequent claims in the proof regarding $i_1, i_2, j_1, j_2$ agreeing on some coordinates follow from this property.

Let $S_1$ be the set of coordinates where $i_1$ and $i_2$ agree. Note that $j_1$ and $j_2$ must also agree in all coordinates in $S_1$ by the above argument. Let $S_2 \subset [q]\setminus S_1$ be the coordinates among the remaining where $i_1$ and $j_1$ agree. Finally, let $S_3 = [q]\setminus (S_1\cup S_2)$. All coordinates in $S_3$ of $i_1$ and $j_2$ must agree. Similarly as before, note that $i_2$ and $j_2$ agree on all coordinates in $S_2$ and $i_2$ and $j_1$ agree on all coordinates in $S_3$. We can rewrite $i_1 = (a,b,c), i_2 = (a,e,f), j_1 = (g, b, f), j_2 = (g, e, c)$ where $a=(a_\ell),g=(g_\ell) $ with $\ell \in S_1$, $b=(b_\ell),e=(e_\ell) $ with $\ell \in S_2$ and $c=(c_\ell),f=(f_\ell) $ with $\ell \in S_3$.

%



First we show that the contribution of the terms where $i_1=i_2$ or $i_1=j_2$ is bounded by $\frac{2\|A_u\|_2^2\|B_{u'}\|_2^2}{m}$, where $A_u$ is the $u$th column of $A$ and $B_{u'}$ is the $u'$th column of $B$. Indeed, consider the case $i_1=i_2$. As observed before, we must have $j_1=j_2$ to get a non-zero contribution. Note that if $t_1\ne t_2$, we always have $\delta_{t_1, i_1}\delta_{t_2, i_2} = 0$ as $H(i_1)$ cannot be equal to both $t_1$ and $t_2$. Thus, for fixed $i_1=i_2, j_1=j_2$,
 \begin{align*}
 & ~ \E \left[\sum_{t_1, t_2 = 1}^m  S(i_1)S(i_2)S(j_1)S(j_2) \cdot
\delta_{t_1,i_1}\delta_{t_1,j_1}\delta_{t_2,i_2}\delta_{t_2,j_2} \cdot
 A_{i_1,u}A_{i_2, u}B_{j_1, u'}B_{j_2, u'}\right] \\
 = & ~ \E\left[\sum_{t_1=1}^m \delta_{i_1,t_1}^2 \delta_{j_1,t_1}^2 A_{i_1, u}^2 B_{j_1, u'}^2\right] \\
 = & ~ \frac{A_{i_1,u}^2 B_{j_1, u'}^2}{m}
\end{align*}
Summing over all possible values of $i_1, j_1$, we get the desired bound of $\frac{\|A_u\|_2^2\|B_{u'}\|_2^2}{m}$. The case $i_1=j_2$ is analogous.

Next we compute the contribution of the terms where $i_1\ne i_2, j_1, j_2$ i.e., there are at least 3 distinct numbers among $i_1, i_2, j_1, j_2$. Notice that $\E[\delta_{t_1,i_1}\delta_{t_1,j_1}\delta_{t_2,i_2}\delta_{t_2,j_2}] \le \frac{1}{m^3}$ because the $\delta_{t,i}$'s are 3-wise independent. For fixed $i_1, j_1, i_2, j_2$, there are $m^2$ choices of $t_1, t_2$ so the total contribution to the expectation from terms with the same $i_1, j_1, i_2, j_2$ is bounded by $m^2\cdot \frac{1}{m^3}\cdot |A_{i_1,u}A_{i_2, u}B_{j_1, u'}B_{j_2, u'}|=\frac{1}{m}|A_{i_1,u}A_{i_2, u}B_{j_1, u'}B_{j_2, u'}|$.

Therefore,
\begin{align*}
& ~\E[((C - A^\top B)_{u,u'})^2]\\
\le & ~ \frac{2\|A_u\|_2^2\|B_{u'}\|_2^2}{m}+\frac{1}{m}\sum_{\textnormal{partition }S_1,S_2,S_3} \sum_{a,g,b,e,c,f} |A_{(a,b,c),u}B_{(g,b,f),u'}A_{(a,e,f),u}B_{(g,e,c),u'}|\\
\le & ~ \frac{2\|A_u\|_2^2\|B_{u'}\|_2^2}{m} + \frac{3^q}{m}\sum_{a,b,c,g,e,f} |A_{(a,b,c),u}B_{(g,b,f),u'}A_{(a,e,f),u}B_{(g,e,c),u'}|\\
\le & ~ \frac{2\|A_u\|_2^2\|B_{u'}\|_2^2}{m} + \frac{3^q}{m}\sum_{g,e,f} \biggl(\sum_{a,b,c}A_{(a,b,c),u}^2\biggr)^{1/2}\biggl(\sum_{a,b,c}B_{(g,b,f),u'}^2 A_{(a,e,f),u}^2 B_{(g,e,c),u'}^2\biggr)^{1/2}\\
= & ~ \frac{2\|A_u\|_2^2\|B_{u'}\|_2^2}{m} + \frac{3^q\|A_u\|}{m}\sum_{g,e,f} \biggl(\sum_{b}B_{(g,b,f),u'}^2\biggr)^{1/2} \biggl(\sum_{a,c}A_{(a,e,f),u}^2 B_{(g,e,c),u'}^2\biggr)^{1/2}\\
\le & ~ \frac{2\|A_u\|_2^2\|B_{u'}\|_2^2}{m} + \frac{3^q\|A_u\|}{m}\sum_{e} \biggl(\sum_{b,g,f}B_{(g,b,f),u'}^2\biggr)^{1/2} \biggl(\sum_{a,c,g,f}A_{(a,e,f),u}^2 B_{(g,e,c),u'}^2\biggr)^{1/2}\\
= & ~ \frac{2\|A_u\|_2^2\|B_{u'}\|_2^2}{m} + \frac{3^q\|A_u\|\cdot\|B_{u'}\|}{m}\sum_{e} \biggl(\sum_{a,f}A_{(a,e,f),u}^2\biggr)^{1/2} \biggl(\sum_{g,c}B_{(g,e,c),u'}^2\biggr)^{1/2}\\
\le & ~ \frac{2\|A_u\|_2^2\|B_{u'}\|_2^2}{m} + \frac{3^q\|A_u\|\cdot \|B_{u'}\|}{m}\biggl(\sum_{a,e,f}A_{(a,e,f),u}^2\biggr)^{1/2} \biggl(\sum_{g,e,c}B_{(g,e,c),u'}^2\biggr)^{1/2}\\
= & ~ \frac{(2+3^q)\|A_u\|_2^2\|B_{u'}\|_2^2}{m},
\end{align*}
where the second inequality follows from the fact that there are at most $3^q$ partitions of $[q]$ into 3 sets. The other inequalities are from Cauchy-Schwarz.

Combining the above bounds, we have
$
\E[((C - A^\top B)_{u,u'})^2] \leq \frac{(2+3^q)\|A_u\|_2^2\|B_{u'}\|_2^2}{m}.
$
For $m\ge (2+3^q)/(\eps^2\delta)$, by the Markov inequality, $\|A^\top S^\top  S B - A^\top  B\|_F^2 \le \eps^2  \|A\|_F^2 \|B\|_F^2$ with probability $1-\delta$.\qed
\end{proof}

The second lemma proves that the subspace embedding property follows from the approximate matrix product property.

\begin{lemma}[Oblivious subspace embeddings]\label{lem:ose}
Consider a fixed $k$-dimensional subspace $V\subset \R^{n_1n_2\cdots n_q}$. If $m \geq k^2 (2+3^q)/(\eps^2\delta)$, then with probability at least $1-\delta$, $\|Sx\|_2 = (1\pm\eps)\|x\|_2$ simultaneously for all $x\in V$.
\end{lemma}
\begin{proof}
Let $B$ be a $(n_1n_2\cdots n_q) \times k$ matrix whose columns form an orthonormal basis of $V$. Thus, we have $B^\top  B = I_k$ and $\|B\|_F^2 = k$. The condition that $\|Sx\|_2 = (1\pm\eps)\|x\|_2$ simultaneously for all $x\in V$ is equivalent to the condition that the singular values of $SB$ are bounded by $1\pm \eps$. By Lemma~\ref{lem:mp}, for $m\geq (2 + 3^q)/((\eps/k)^2\delta)$, with probability at least $1-\delta$, we have
$$
\|B^\top  S^\top  S B -  B^\top  B\|_F^2 \le (\eps/k)^2 \|B\|_F^4 = \eps^2
$$
Thus, we have $\|B^\top S^\top  S B -  I_k\|_2 \le \|B^ \top S^\top  S B -  I_k\|_F \le \eps$. In other words, the squared singular values of $SB$ are bounded by $1\pm \eps$, implying that the singular values of $SB$ are also bounded by $1\pm \eps$. Note that $\|A\|_2$ for a matrix $A$ denotes its operator norm.\qed
\end{proof}

\section{Missing Proofs}\label{Sec:missing_proofs}
\subsection{Proofs for Tensor Product Least Square Regression}
\label{sec:proof_norm_two}

\restate{them:norm_two_reg}
\begin{proof}
It is easy to see that
$$
\|(A_1\otimes A_2\otimes \cdots \otimes A_q)x-b\|_2=\left\|\begin{bmatrix}(A_1\otimes A_2\otimes \cdots \otimes A_q)& b\end{bmatrix} \begin{bmatrix}x \cr  -1\end{bmatrix} \right\|_2,
$$
and identifying $$
y=\begin{bmatrix}(A_1\otimes A_2\otimes \cdots \otimes A_q) & b\end{bmatrix} \begin{bmatrix}x \cr  -1\end{bmatrix} \in \R^{n_1n_2\cdots n_q}
$$
and $y$ is a vector of a subspace $V\subset \R^{n_1 n_2\cdots n_q}$ with dimension at most $d_1 d_2\cdots d_q+1$, we can use Lemma  \ref{lem:ose} to conclude that
\begin{align*}
\Pr \left[ \left|\left\|S y\right\|_2-\|y\|_2\right| \leq  \epsilon \|y\|_2  \right] \geq 1-\delta 
\end{align*}
when $m=(d_1d_2\cdots d_q +1)^2 (2+3^q)/(\eps^2\delta)$. 
%

Thus we have
\begin{align*}
\|(A_1\otimes A_2\otimes \cdots \otimes A_q)\widetilde { x}-b\|_2	 \leq \frac{1}{1-\epsilon}
\|S(A_1\otimes A_2\otimes \cdots \otimes A_q) \widetilde { x}-Sb\|_2	
	\end{align*}
	and 
	\begin{align*}
\|S(A_1\otimes A_2\otimes \cdots \otimes A_q)  { x}-Sb\|_2	 \leq (1+\epsilon)
\|(A_1\otimes A_2\otimes \cdots \otimes A_q) { x}-b\|_2	
	\end{align*}
	hold with probability at least  $1-\delta$. Then using a union bound, we have
	\begin{align*}
& ~ \|(A_1\otimes A_2\otimes \cdots \otimes A_q) \widetilde { x}-b\|_2	 \\
 \leq & ~\frac{1}{1-\epsilon} \|S(A_1\otimes A_2\otimes \cdots \otimes A_q) \widetilde { x}-Sb\|_2	\\
\leq & ~ \frac{1}{1-\epsilon} \|S(A_1\otimes A_2\otimes \cdots \otimes A_q) { x}-Sb\|_2	\\
\leq & ~\frac{1+\epsilon}{1-\epsilon} \|(A_1\otimes A_2\otimes \cdots \otimes A_q)  { x}-b\|_2	
	\end{align*}
holds with probability at least $1-2\delta $. 
\end{proof}

\restate{them:norm_two_reg_positive}
\begin{proof}	
The proof of Theorem.~\ref{them:norm_two_reg_positive} is similar to the proof of theorem~\ref{them:norm_two_reg}. Denote $\tilde{x} = \min_{x\geq 0} \|S\mathcal{A}x - Sb\|_2$ and $x^* = \min_{x\geq 0}\|\mathcal{A}x - b\|_2$. Using Lemma.~\ref{lem:ose}, we have:
\begin{align} 
\|\mathcal{A}\tilde{x} - b\|_2 \leq \frac{1}{1-\epsilon}\|S\mathcal{A}\tilde{x} - Sb\|_2,
\end{align} with probability at least $1-\delta$, and 
\begin{align}
\|S\mathcal{A}{x}^* - Sb\|_2 \leq (1+\epsilon)\|\mathcal{A}{x}^* - b\|_2,
\end{align} with probability at least $1-\delta$.
Hence applying a union bound we have:
\begin{align}
&  ~ \|\mathcal{A}\tilde{x} - b\|_2 \\
\leq & ~ \frac{1}{1-\epsilon}\|S\mathcal{A}\tilde{x} - Sb\|_2\nonumber\\
\leq & ~\frac{1}{1-\epsilon}\|S\mathcal{A}{x}^* - Sb\|_2\nonumber\\
\leq & ~\frac{1+\epsilon}{1-\epsilon}\|\mathcal{A}{x}^* - b\|_2,
\end{align}  with probability at least $1-2\delta$. 
\end{proof}


\subsection{Proofs for P-Splines}
	
\restate{lem:reg}
\begin{proof}
Let $\hat{A}\in\R^{(n+d)\times d}$ have orthonormal columns
with $\range(\hat{A} )=\range( \twomat{A}{\sqrt{\lambda} L} )$.
(An explicit expression for one such $\hat{A}$ is given below.)
Let $\hat{b} \equiv \twomat{b}{0_d}$.
We have
\begin{equation}\label{eq ridge2}
\min_{y\in\R^d} \norm{\hat{A}y - \hat{b}}
\end{equation}
equivalent to $\left\|b- A x \right\|_2^2+\lambda \|Lx\|_2^2$,
, in the sense that for any $\hat{A}y\in\range(\hat{A})$, there is $x\in\R^d$ with
$\hat{A}y = \twomat{A}{\sqrt{\lambda} L} x$,
so that
$\norm{\hA y - \hat{b}}^2 = \norm{\twomat{A}{\sqrt{\lambda} L} x - \hat{b}}^2
=\norm{b-Ax}^2 +  \lambda\norm{L x}^2$.
Let $y^* = \argmin_{y\in\R^d} \norm{\hat{A}y - \hat{b}}$,
so that $\hA y^* = \twomat{Ax^*}{\sqrt{\lambda} L x^*}$.

Let $\hat{A} =  \twomat{U_1}{U_2}$, where $U_1\in\R^{n\times d}$ and
$U_2\in\R^{d\times d}$, so that $U_1$ is as in the lemma statement.

We define $\hS$ to be $\left[\begin{smallmatrix} S & 0_{m\times d}\\ 0_{d\times n} & \Iden_d \end{smallmatrix}\right]$ and $\hS$ satisfies Property~(\RN{1}) and (\RN{2}) of Lemma~\ref{lem:reg}.


Using $\norm{U_1^\top S^\top S U_1 - U_1^\top U_1}\le 1/4$, with constant probability
\begin{equation}\label{eq embed2}
\norm{\hA^\top \hS^\top \hS \hA - \Iden_d}
	= \norm{U_1^\top S^\top S U_1 + U_2^\top U_2 - \Iden_d}
	= \norm{U_1^\top S^\top S U_1 - U_1^\top U_1}
	\le 1/4.
\end{equation}
Using the normal equations for Eq.~\eqref{eq ridge2}, we have
\begin{equation*}
0
	= \hA^\top (\hat{b} - \hA y^*)
	= U_1^\top (b-Ax^*) - \sqrt{\lambda} U_2^\top x^*,
\end{equation*}
and so
\[
\hA^\top \hS^\top \hS (\hat{b} - \hA y^*)
	= U_1^\top S^\top S (b-Ax^*) - \sqrt{\lambda} U_2^\top x^*
	= U_1^\top S^\top S (b-Ax^*)  - U_1^\top (b-Ax^*).
\]
Using Property (\RN{2}) of Lemma~\ref{lem:reg}, with constant probability
\begin{align}
 & ~ \norm{\hA^\top \hS^\top \hS (\hat{b} - \hA y^*)} \notag \\
	   = & ~ \norm{U_1^\top S^\top S (b-Ax^*)  - U_1^\top (b-Ax^*)} \notag \\ 
	   \le & ~ \sqrt{\eps \OPT/2} \notag \\ 
	   = & ~\sqrt{\eps/2} \norm{\hat{b} - \hA y^*}. \label{eq prod2}
\end{align}
It follows by a standard result  from \eqref{eq embed2} and \eqref{eq prod2} that the solution
$\tilde{y} \equiv \argmin_{y\in\R^d} \norm{\hS(\hA y - \hat{b})}$ has
$\norm{\hA\tilde{y} - \hat{b}}\le (1+\eps)\min_{y\in\R^d} \norm{\hat{A}y - \hat{b}}$,
and therefore that $\tilde{x}$ satisfies the claim of the theorem.

For convenience we give the proof of the standard result: \eqref{eq embed2} implies that
$\hA^\top \hS^\top \hS\hA$ has smallest singular value at least $3/4$. The normal equations for the unsketched
and sketched problems are
$$\hA^\top(\hat{b} - \hA y^*) = 0 = \hA^\top \hS^\top \hS (\hat{b} - \hA \tilde{y}).$$
The normal equations for the unsketched case imply
$\norm{\hA \ty - \hb}^2 = \norm{\hA(\ty - y^*)}^2 + \norm{\hb - \hA y^*}^2$,
so it is enough to show that $\norm{\hA(\ty - y^*)}^2 = \norm{\ty - y^*}^2 \le \eps\OPT$.
We have
\begin{align*}
(3/4) \norm{\ty - y^*}
	   & \le \norm{\hA^\top \hS^\top \hS \hA(\ty - y^*)} & \text{by~Eq.~\eqref{eq embed2}}
	\\ & = \norm{\hA^\top \hS^\top \hS \hA(\ty - y^*) - \hA^\top \hS^\top \hS (\hat{b} - \hA \tilde{y})} &  \text{~by~Normal Equation}
	\\ & = \norm{\hA^\top \hS^\top \hS(\hb - \hA y^*)} &
	\\ & \le \sqrt{\eps\OPT/2}  & \text{by~Eq.~\eqref{eq prod2}},
\end{align*}
so that $\norm{\ty - y^*}^2 \le (4/3)^2\eps\OPT/2\le \eps\OPT$. The lemma follows.

\end{proof}

The following lemma computes the statistical dimension $\sd_\lambda(A,L)$ that will be used for computing the  number of rows of sketching matrix $S$. 
\begin{lemma}\label{lem U1 size}
For $U_1$ as in Lemma~\ref{lem:reg}, $\normF{U_1}^2 = \sd_\lambda(A,L)  =  \sum_i 1/(1+ \lambda/\gamma_i^2)+d-p$, where
$A$ has singular values $\sigma_i$. Also $\norm{U_1} = \max\{1/\sqrt{1+ \lambda/\gamma_1^2},\, 1\}$.
\end{lemma}
\begin{proof}
Suppose we have the GSVD of $(A,\, L)$. Let 
$$
D\equiv \begin{bmatrix}\Sigma^\top \Sigma + \lambda \Omega^\top \Omega & { 0}_{p\times (n-p)}\cr { 0}_{(n-p)\times p}& I_{d-p}
\end{bmatrix}^{-1/2}.
$$
Then
$$\hA = \twomat{U\begin{bmatrix}\Sigma& { 0}_{p\times (n-p)}\cr { 0}_{(n-p)\times p}& I_{d-p}
\end{bmatrix} D}{\sqrt{\lambda} V\begin{bmatrix}\Omega  &{
0}_{p \times (n-p)}\end{bmatrix} D}
$$
has $\hA^\top \hA = \Iden_d$, and for given $x$, there is $y=D^{-1} RQ^\top x$ with
$\hA y = \twomat{A}{\sqrt{\lambda} L} x$.
We have $\normF{U_1}^2 = \left\|{U\begin{bmatrix}\Sigma& { 0}_{p\times (n-p)}\cr { 0}_{(n-p)\times p}& I_{d-p}
\end{bmatrix} D}\right\|_F^2 =  \left\|{\begin{bmatrix}\Sigma& { 0}_{p\times (n-p)}\cr { 0}_{(n-p)\times p}& I_{d-p}
\end{bmatrix} D}\right\|_F^2=  \sum_{i=1}^p 1/(1+ \lambda/\gamma_i^2)+d-p$
as claimed.  
\end{proof}

\restate{thm:size_of_S}
\begin{proof}
Recall that $\sd_\lambda(A,L) = \normF{U_1}^2$.
Sparse embedding distributions satisfy the bound for approximate matrix multiplication
$$\normF{W^\top S^\top S H - W^\top H} \le C\normF{W}\normF{H}/\sqrt{m},$$
for a constant $C$ \citep{cw13,mm13,nn13}; this is also true of OSE matrices.
We set $W=H=U_1$ and use $\norm{X} \le \normF{X}$ for all $X$ and
$m\ge K \normF{U_1}^4$ to obtain Property~(\RN{1}) of Lemma~\ref{lem:reg}, and set $W=U_1$, $H=b - Ax^*$
and use $m\ge K \normF{U_1}^2/\eps$ to obtain Property~(\RN{2}) of Lemma~\ref{lem:reg}. (Here the bound is slightly stronger
than Property~(\RN{2}), holding for $\lambda=0$.) With Property~(\RN{1}) and Property~(\RN{2}), the
claim for $\tilde{x}$ from a sparse embedding follows using Lemma~4.1.
\end{proof}

\subsection{Proofs for Tensor Product $\ell_1$ Regression}

\restate{lem:lpl2}
\begin{proof}
This lemma is similar to arguments in \cite{CDMMMW}, we simply adjust notation and parameters for completeness.
Applying Theorem~\ref{thm:jlmain}, we have that
with probability at least $1-\prod_{i=1}^q (n_i/w_i) \delta$, for all $x \in \mathbb{R}^r$, if we consider $y = \cA x$ and
write $y^\top = [z_1^\top, z_2^\top, \ldots, z_{\prod_{i=1}^q n_i/w_i}]^\top $, then for all $i \in [\prod_{i=1}^q n_i/w_i]$,  
\begin{eqnarray*}
\sqrt{\textstyle\frac 1 2}\|{z_i}\|_2\leq \|S_i z_i \|_2\leq 
\sqrt{\textstyle\frac 3 2}\|{z_i}\|_2,
\end{eqnarray*}
where $S_{i} \in \R^{m_i \times \prod_{i=1}^q w_i}$.  In the following, suppose $m_i=t$. By relating the $2$-norm and the $p$-norm, for $1 \leq p \leq 2$, we have
$$\|{S_i z_i}\|_p
	\le t^{1/p-1/2}\|{Sz_i}\|_2
	\le t^{1/p-1/2} \sqrt{\textstyle\frac32} \|{z_i}\|_2
	\le t^{1/p-1/2} \sqrt{\textstyle\frac32} \|{z_i}\|_p,$$
and similarly,
$$
\| {S_i z_i}\|_p
	\ge \|{S_iz_i}\|_2
	\ge \sqrt{\textstyle\frac12}\|{z_i}\|_2
	\ge \sqrt{\textstyle\frac12} w^{1/2-1/p}\|{z_i}\|_p,\, w=\prod_{j=1}^q w_j. 
$$
If $p > 2$, then 
$$
\|{S_i z_i}\|_p
	\le \|{S_i z_i}\|_2
	\le \sqrt{\textstyle\frac32} \|{z_i}\|_2
	\le \sqrt{\textstyle\frac32} w^{1/2 - 1/p} \|{z_i}\|_p,$$
and similarly,
$$
\| {S_i z_i}\|_p
	\ge t^{1/p-1/2}\|{S_i z_i}\|_2
	\ge t^{1/p-1/2} \sqrt{\textstyle\frac12} \|{z_i}\|_2
 	\ge t^{1/p-1/2} \sqrt{\textstyle\frac12}  \|{z_i}\|_p.
$$
Since $\|{\cA x}\|_p^p = \|{y}\|_p^p = \sum_i \|{z_i}\|_p^p$
and $\|{ S \cA x}\|_p^p = \sum_i \|{S_iz_i}\|_p^p$,
for $p\in [1,2]$ we have with probability $1-\prod_{i=1}^q (n_i/w_i) \delta$
\[
\sqrt{\textstyle\frac12} w^{1/2-1/p}\|{\cA x}\|_p
	\le \|{S \cA x}\|_p \le \sqrt{\textstyle\frac32} t^{1/p-1/2} \|{\cA x}\|_p,
\]
and for $p\in [2,\infty)$ with probability $1-\prod_{i=1}^q (n_i/w_i) \delta$
\[
\sqrt{\textstyle\frac12} t^{1/p-1/2} \|{\cA x}\|_p
	\le \|{S \cA x}\|_p \le \sqrt{\textstyle\frac32} w^{1/2 - 1/p} \|{\cA x}\|_p.
\]
In either case,
\begin{equation}\label{eqn:first}
\|{\cA x}\|_p \le \gamma_p \|{S \cA x}\|_p \le \sqrt{3} (tw)^{|1/p-1/2|}\|{\cA x}\|_p.
\end{equation}



We have, from the definition
of an $(\alpha,\beta,p)$-well-conditioned basis, that 
\begin{eqnarray}\label{eqn:second}
\|S \cA U \|_p \leq \alpha 
\end{eqnarray}
and for all $x \in \mathbb{R}^d$,
\begin{eqnarray}\label{eqn:third}
\|x\|_q \leq \beta \|S \cA Ux\|_p.
\end{eqnarray}

Combining (\ref{eqn:first}) and (\ref{eqn:second}), we have that with probability at least $1-\prod_{i=1}^q (n_i/w_i) \delta$,
\begin{eqnarray*}
\|\cA U/(r\gamma_p) \|_p \leq \sum_i \|\cA U_i/r\gamma_p \|_p
	\leq \sum_i  \|S \cA U_i/r\|_p
	\leq \alpha.
\end{eqnarray*}
Combining (\ref{eqn:first}) and (\ref{eqn:third}), we have that with probability at least $1-\prod_{i=1}^q (n_i/w_i) \delta$,
for all $x \in \mathbb{R}^r$,
\begin{eqnarray*}
\|x\|_q \leq \beta \|S \cA U x\|_p
	\leq \beta \sqrt{3} r (tw)^{|1/p-1/2|} \|\cA U\frac{1}{r\gamma_p}  x\|_p.
\end{eqnarray*}
Hence $\cA U/(r\gamma_p)$ is an 
$(\alpha, \beta \sqrt{3} r (tw)^{|1/p-1/2|} , p)$-well-conditioned
basis. 
\end{proof}
 

\bigskip


\restate{thm:lp_running}
\begin{proof}
For notational simplicity, let us denote $n_{[q_1]} = \prod_{i=1}^{q_1} n_i$, $n_{[q]\backslash [q_1]} = \prod_{i=q_1+1}^{q} n_1 $, $d_{[q_1]} = \prod_{i=1}^{q_1} d_i$, and $d_{[q]\backslash [q_1]} = \prod_{i=q_1+1}^{q} d_i$.
	For any row-block $A_{i_1}^{1}\otimes \ldots \otimes A_{i_q}^{(q)}$, computing $S_{i1i2...i_q}(A_{i_1}^{1}\otimes \ldots \otimes A_{i_q}^{(q-1)})$ takes $O(d(\sum_{k=1}^q nnz(A_{i_k}^{(k)})) + dqm\log(m))$ (see Sec~\ref{sec:background}). Hence for $S\cA$, it takes:
	\begin{align*}
	\left( d \sum_{k=1}^q \nnz(A_k) \prod_{i \in [q] \backslash \{k\} }^q  n_i/w_i  \right) + \left( dqm\log(m) \prod_{i=1}^{q} n_i/w_i \right)
	\end{align*}
where $S \in \R^{(m\prod_{i=1}^q(n_i/w_i))  \times \prod_{i=1}^q w_i} $ and  $m \geq 100\prod_{i=1}^q d_i^2 (2+3^q)/\eps^2 = O(\poly(d/\epsilon))$. We need to compute an orthogonal factorization $S\cA=QR_\cA $ in $O(q m d^2)$ and then compute $U=R_\cA^{-1}$ in $O(d^3)$ time. Hence the total running time of Algorithm {\sf Condition}($\cA$) is $O(q m d^2+d^3)$. Thus the total running time of computing $S\cA$ and Condition(\cA) is 
\begin{align*}
O\left( \left( \sum_{k=1}^q \nnz (A_k) \prod_{i\in [q] \backslash \{k\}}^q n_i/w_i \right) + \left(\prod_{i=1}^q n_i/w_i \right)\poly(d/\epsilon) + q m d^2+d^3 \right), 
\end{align*}

We will compute $UG$ in $O( d^2 \log n )$ time. We compute $\widetilde E=E (A_{q_1+1}\otimes\ldots\otimes A_{q})^T $ in $O(d n_{[q]\backslash [q_1]})$ time. 

Then we can compute  $R (A_1 \otimes \cdots \otimes A_{q_1})  \widetilde E_j$ in $O(n_{[q_1]}d_{[q_1]}\log n + d_{[q_1]}n_{[q] \backslash[q_1]} \log n )$ time.

Since computation of the median $\lambda_i$ takes $O(\log n)$ time, computing all $\lambda_i$ and then $\lambda_e$ takes $O(n_{[q] \backslash[q_1]}\log n )$ time. 

As $\cA UG$ has $O(\log n)$ columns, we need to compute $\lambda_e$ for each $\cA UG$ using the above procedure and hence it takes in total $O( d(n_{[q_1]}+n_{[q] \backslash[q_1]}) \log^2 n)$ time. 

Sampling a column of $\cA UG$ using $\lambda_e$ takes $O(\log n)$ time, sampling an entry in $M$ takes in total $O(n_{[q_1]} + n_{[q] \backslash[q_1]})$ time. 

Since we need $\sqrt{\prod_{k=1}^q w_k} \poly(  r )$ samples to select rows,  the running time is $  d (n_{[q_1]} + n_{[q]\backslash [q_1] } ) \log^2 n \cdot \sqrt{\prod_{k=1}^q w_k} \poly(  r )$.

Now for simplicity, we set $q = 2$, $n_i = n_0$ for $i\in [2]$. Note that it is optimal to choose $w_i = w$ for $i\in[2]$. Substituting $q = 2, n_i =n_0$ and $w_i = w$, we that the total running time of Alg.~\ref{alg:l1tensorregression}:
\begin{align*}
O\left( d w^{-1} n_0 (\nnz(A_1) + \nnz(A_2)) + w^{-2} n_0^2 \poly(d/\epsilon) + w n_0 \poly(d)\log(n) \right).
\end{align*}
For dense $A_1$ and $A_2$, $\nnz(A_1)+\nnz(A_2) = O(n_0)$ time, and so ignoring $\poly$ and $\log$ terms that do not depend on $n_0$,  the total running time can be simplified to:
\begin{align*}
O( w^{-1}n_0^2 + wn_0).
\end{align*} 
Setting $w = \sqrt{n_0}$, we can minimize the above running time to $O(n_0^{3/2})$, which is faster than the $n_0^2$ time for solving the problem by forming $A_1\otimes A_2$.
\end{proof}


\end{document}